\documentclass[%
 reprint,
 amsmath,amssymb,
 aps,
]{revtex4-1}
\usepackage{amsmath}
\usepackage{booktabs}
\usepackage{array}
\usepackage{amsthm}
\usepackage{float}
\usepackage{graphicx}
\usepackage{dcolumn}
\usepackage{bm}
\usepackage{hyperref}
\usepackage{graphicx}
\usepackage{setspace}
\usepackage{amsmath}
\usepackage{xcolor}
\usepackage{multirow}
\usepackage{float}
\usepackage{caption}
\usepackage{multirow}
\usepackage{setspace}
\usepackage{appendix}
\usepackage{bm}
\usepackage{makecell}
\usepackage{array}
\usepackage{tabularx}
\usepackage{booktabs}
\usepackage{makecell}
\usepackage{threeparttable}

\usepackage{float}
\usepackage[english]{babel}
\usepackage{caption}
\usepackage[figuresright]{rotating}

\RequirePackage[normalem]{ulem} 
\RequirePackage{color}\definecolor{RED}{rgb}{1,0,0}\definecolor{BLUE}{rgb}{0,0,1} 
\providecommand{\DIFaddbegin}{} 
\providecommand{\DIFaddend}{} 
\providecommand{\DIFdelbegin}{} 
\providecommand{\DIFdelend}{} 
\providecommand{\DIFaddbeginFL}{} 
\providecommand{\DIFaddendFL}{} 
\providecommand{\DIFdelbeginFL}{} 
\providecommand{\DIFdelendFL}{} 
\newcommand{\DIFscaledelfig}{0.5}
\RequirePackage{settobox} 
\RequirePackage{letltxmacro} 
\newsavebox{\DIFdelgraphicsbox} 
\newlength{\DIFdelgraphicswidth} 
\newlength{\DIFdelgraphicsheight} 
\LetLtxMacro{\DIFOincludegraphics}{\includegraphics} 
\newcommand{\DIFaddincludegraphics}[2][]{{\color{blue}\fbox{\DIFOincludegraphics[#1]{#2}}}} 
\newcommand{\DIFdelincludegraphics}[2][]{
\sbox{\DIFdelgraphicsbox}{\DIFOincludegraphics[#1]{#2}}
\settoboxwidth{\DIFdelgraphicswidth}{\DIFdelgraphicsbox} 
\settoboxtotalheight{\DIFdelgraphicsheight}{\DIFdelgraphicsbox} 
\scalebox{\DIFscaledelfig}{
\parbox[b]{\DIFdelgraphicswidth}{\usebox{\DIFdelgraphicsbox}\\[-\baselineskip] \rule{\DIFdelgraphicswidth}{0em}}\llap{\resizebox{\DIFdelgraphicswidth}{\DIFdelgraphicsheight}{
\setlength{\unitlength}{\DIFdelgraphicswidth}
\begin{picture}(1,1)
\thicklines\linethickness{2pt} 
{\color[rgb]{1,0,0}\put(0,0){\framebox(1,1){}}}
{\color[rgb]{1,0,0}\put(0,0){\line( 1,1){1}}}
{\color[rgb]{1,0,0}\put(0,1){\line(1,-1){1}}}
\end{picture}
}\hspace*{3pt}}} 
} 
\LetLtxMacro{\DIFOaddbegin}{\DIFaddbegin} 
\LetLtxMacro{\DIFOaddend}{\DIFaddend} 
\LetLtxMacro{\DIFOdelbegin}{\DIFdelbegin} 
\LetLtxMacro{\DIFOdelend}{\DIFdelend} 
\DeclareRobustCommand{\DIFaddbegin}{\DIFOaddbegin \let\includegraphics\DIFaddincludegraphics} 
\DeclareRobustCommand{\DIFaddend}{\DIFOaddend \let\includegraphics\DIFOincludegraphics} 
\DeclareRobustCommand{\DIFdelbegin}{\DIFOdelbegin \let\includegraphics\DIFdelincludegraphics} 
\DeclareRobustCommand{\DIFdelend}{\DIFOaddend \let\includegraphics\DIFOincludegraphics} 
\LetLtxMacro{\DIFOaddbeginFL}{\DIFaddbeginFL} 
\LetLtxMacro{\DIFOaddendFL}{\DIFaddendFL} 
\LetLtxMacro{\DIFOdelbeginFL}{\DIFdelbeginFL} 
\LetLtxMacro{\DIFOdelendFL}{\DIFdelendFL} 
\DeclareRobustCommand{\DIFaddbeginFL}{\DIFOaddbeginFL \let\includegraphics\DIFaddincludegraphics} 
\DeclareRobustCommand{\DIFaddendFL}{\DIFOaddendFL \let\includegraphics\DIFOincludegraphics} 
\DeclareRobustCommand{\DIFdelbeginFL}{\DIFOdelbeginFL \let\includegraphics\DIFdelincludegraphics} 
\DeclareRobustCommand{\DIFdelendFL}{\DIFOaddendFL \let\includegraphics\DIFOincludegraphics} 

\begin{document}

\theoremstyle{remark}
\newtheorem{definition}{\indent Definition}
\newtheorem{lemma}{\indent Lemma}
\newtheorem{theorem}{\indent Theorem}
\newtheorem{proposition}{Proposition}
\newtheorem{construction}{Construction}
\newtheorem{corollary}{\indent Corollary}
\newtheorem{example}{Example}
\newtheorem{remark}{Remark}
\newcommand{\Span}{\mathrm{span}}
\def\QEDclosed{\mbox{\rule[0pt]{1.3ex}{1.3ex}}}
\def\QED{\QEDclosed}
\def\proof{\indent{\em Proof}.}
\def\endproof{\hspace*{\fill}~\QED\par\endtrivlist\unskip}

\preprint{APS/123-QED}

\title{Activate genuine nonlocality from distinguishable sets in tripartite systems}

\author{Hui-Juan Zuo$^{1,2}$}\email{Contact author: huijuanzuo@163.com}
\author{Ying-Ying Lu$^{1}$}
\author{Shao-Ming Fei$^{3}$}

\affiliation{
$^{1}$School of Mathematical Sciences, Hebei Normal University, Shijiazhuang, China\\
$^{2}$State key Laboratory of Networking and Switching Technology, Beijing University of Posts and Telecommunications, Beijing, China\\
$^{3}$School of Mathematical Sciences, Capital Normal University,
Beijing, China
}%


\begin{abstract}
A set of orthogonal quantum states in multipartite systems is of genuine nonlocality if it is locally indistinguishable in every bipartition. If it is locally reducible when the parties are separated, we say that it has genuine nonlocality of type~\uppercase\expandafter{\romannumeral 1}; otherwise, it has genuine nonlocality of type~\uppercase\expandafter{\romannumeral 2}. For a locally distinguishable set without local redundancy, if there exist some orthogonality preserving local measurements such that each outcome leads to a locally indistinguishable set, then we say that it exhibits the activation of nonlocality. We activate type-\uppercase\expandafter{\romannumeral 1} and type-\uppercase\expandafter{\romannumeral 2} genuine nonlocality of orthogonal product state sets in tripartite systems. In particular, we tackle the local irredundancy problem with partial trace operation and $p$-ary numeral systems to significantly simplify the proofs. Our results also address the open question raised by S. Bandyopadhyay \textit{et al.}[\href{https://link.aps.org/doi/10.1103/PhysRevA.104.L050201}{Phys. Rev. A \textbf{104}, L050201 (2021)}]. Furthermore, we observe the activation of hidden genuine nonlocality in multipartite systems, which highlights the applications of nonlocality based on state discrimination in different practical scenarios.

\begin{description}
\item[PACS numbers]
03.65.Ud, 03.67.Mn
\end{description}
\end{abstract}

\pacs{Valid PACS appear here}
\maketitle


\section{\label{sec:level1}Introduction\protect}

Quantum nonlocality is a magical phenomenon in quantum information theory that has attracted much attention. Since the Bell nonlocality \cite{Bell} was proposed, it has been believed that entanglement is a necessary condition for quantum nonlocality. Until 1999, Bennet \textit{et al.} \cite{Bennett1999} first proposed the phenomenon of  ``quantum nonlocality without entanglement" based on a set of locally indistinguishable orthogonal product bases. From then on, quantum nonlocality from the perspective of local discrimination has been extensively and deeply studied \cite{Walgate2000,Walgate2002,Yu2015,Zhang2016,Wang2017,Zuo2021,Zhen2022,Xu202501,Xu202502}. A set of orthogonal quantum states is {\it locally indistinguishable or nonlocal} if it is not possible to perfectly distinguish the states by any sequence of local operations and classical communications (LOCC). It has utmost applications in distributed quantum protocols such as quantum data hiding and quantum secret sharing \cite{DiVincenzo2003,Chattopadhyay2007,Rahaman2015,WangJT2017,Jiang2018}. As research progresses, people seek to identify the sets of orthogonal quantum states exhibiting stronger nonlocality, such as local irreducibility \cite{Halder2019}, local stability \cite{Li2023,Wang2023,Cao2023,Zhen2025}, local unmarkability \cite{Sen2022,Hsu2023}, genuine nonlocality \cite{Rout2019,Li2021,Rout2021,Xiong2023,Xiong202401,xiong202402,Lu2024,Zhang2025}, strong nonlocality \cite{Yuan2020,Shi2020,Wang2021,He2024} and the strongest nonlocality \cite{Shi202203,LiJC2023,Zhen2024}.

In quantum communication, once the information is encoded in a locally distinguishable set, the complexity of retrieving hidden information depends on the degree of mutual trust among the participants. If their mutual trust changes after secret sharing, the ability of a certain participant to manipulate the complexity becomes particularly critical. This naturally leads to the concept of {\it nonlocality activation}, that is, a distinguishable set without local redundancy can be transformed into a locally indistinguishable set through orthogonality-preserving local measurements (OPLM). Bandyopadhyay \textit{et al.} \cite{Bandyopadhyay2021} proposed the concept of genuinely activable hidden nonlocality and constructed the
sets of orthogonal entangled states with this property in $\mathbb{C}^{2}\otimes\mathbb{C}^{4}$, $\mathbb{C}^{4}\otimes\mathbb{C}^{4}$, $\mathbb{C}^{5}\otimes\mathbb{C}^{5}$, and $\mathbb{C}^{5}\otimes\mathbb{C}^{5}\otimes\mathbb{C}^{5}$. Li \textit{et al.} \cite{Li2022} introduced the classification of genuine hidden nonlocality. For odd-dimensional systems, they obtained the orthogonal product state sets with  genuine hidden nonlocality of type \uppercase\expandafter{\romannumeral 1} in $\otimes_{i=1}^{n} (\mathbb{C}^{d_{i}}\otimes\mathbb{C}^{d_{i}})$ $(d_{i}\geq 11)$. Moreover, they explicitly provided a set of orthogonal product states exhibiting type-\uppercase\expandafter{\romannumeral 2} genuine hidden nonlocality in $\mathbb{C}^{m}\otimes\mathbb{C}^{n}$ $(m, n \geq 7)$. Zhang \textit{et al.} \cite{Zhang2024} constructed orthogonal product state sets with type-\uppercase\expandafter{\romannumeral 1} genuine hidden nonlocality in $\mathbb{C}^{d_{1}}\otimes\mathbb{C}^{d_{2}}$ ($d_{1}\geq 8, d_{2}\geq 10$, with even $d_{1}, d_{2}$) and type-\uppercase\expandafter{\romannumeral 2} genuine hidden nonlocality in $\mathbb{C}^{d}\otimes\mathbb{C}^{2d}$~($d\geq3$). It should be noted that the classification of type-\uppercase\expandafter{\romannumeral 1} and type-\uppercase\expandafter{\romannumeral 2} in Refs.[40,41] depends on the change of the cardinality of the state set before and after the activation. In this paper, however, the classification is based on the two classes of genuine nonlocality about the activated set (see Definition 4). We mainly investigate the special properties of the activated genuinely nonlocal set.

The activation of stronger nonlocal properties for distinguishable sets such as strong nonlocality and local unmarkability, has already attracted much attention in recent years. Ghosh \textit{et al.} \cite{Ghosh2022} successfully activated the local irreducibility of orthogonal product state sets in $\mathbb{C}^{3}\otimes\mathbb{C}^{6}$ and $\mathbb{C}^{2}\otimes\mathbb{C}^{2}\otimes\mathbb{C}^{4}$. Furthermore, they constructed the set of orthogonal product states with genuinely activable strong nonlocality in $\mathbb{C}^{3}\otimes\mathbb{C}^{3}\otimes\mathbb{C}^{6}$ and $\mathbb{C}^{6}\otimes\mathbb{C}^{6}\otimes\mathbb{C}^{6}$. It is worth mentioning that their first three constructions have the smallest dimension. Gupta \textit{et al.} \cite{Gupta2023} proposed the concept of relaxed local redundancy and hierarchically activated special nonlocal properties of orthogonal entangled state sets in $\mathbb{C}^{4}\otimes\mathbb{C}^{8}$, $\mathbb{C}^{2}\otimes\mathbb{C}^{4}$, $\mathbb{C}^{4}\otimes\mathbb{C}^{4}$, and $\mathbb{C}^{4}\otimes(\mathbb{C}^{2})^{\otimes(n-1)}$. Bera \textit{et al.} \cite{Bera2024} classified activable sets based on the number of activating parties and compared the strength of locality from this perspective. They discussed single-party and multi-party activable local sets in $\mathbb{C}^{3}\otimes\mathbb{C}^{2}\otimes\mathbb{C}^{3}$, then generalized it to $\mathbb{C}^{2m+1}\otimes\mathbb{C}^{2}\otimes\mathbb{C}^{2m+1}$. Moreover, they provided a set of locally distinguishable orthogonal states in $\mathbb{C}^{8}\otimes\mathbb{C}^{8}\otimes\mathbb{C}^{8}$, which cannot be activated by single-party but jointly activated by any two parties. Recently, Bhunia \textit{et al.} \cite{Bhunia2026} provided the sets of multipartite states, which are not activable in any bipartition.

In this work, we focus on the activation of hidden genuine nonlocality for distinguishable orthogonal quantum states without entanglement. We introduce several basic concepts and notations in Section 2. In Section 3, two sets of genuinely nonlocal orthogonal product states are presented. Based on these constructions, we activate type-\uppercase\expandafter{\romannumeral 1} genuine nonlocality of distinguishable orthogonal product state sets in $\mathbb{C}^{d_{1}}\otimes\mathbb{C}^{d_{2}}\otimes\mathbb{C}^{d_{3}}$ in Section 4 and type-\uppercase\expandafter{\romannumeral 2} genuine nonlocality in Section 5. Finally, in Section 6, the conclusion is drawn.

\section{Notations and definitions}\label{sec:2}

For any positive integer $d\geq2$, we denote $\mathbb{Z}_{d}=$ \{$0,1,\cdots,d-1$\}. Without loss of generality, $\{|0\rangle,|1\rangle,\cdots,|d-1\rangle\}$ is the computational basis of the quantum system $\mathbb{C}^{d}$ with dimension $d$. For simplicity, we do not normalize the states and denote $|i_{1} \pm i_{2} \pm \cdots \pm i_{n}\rangle=\frac{1}{\sqrt{n}}(|i_{1}\rangle \pm |i_{2}\rangle \pm \cdots \pm |i_{n}\rangle)$ and $|_{m}+_{n}\rangle=\sum_{i=m}^{n}|i\rangle$ where $m<n$.
We only consider pure states and positive operator-valued measures (POVM).

\begin{definition}\cite{Walgate2002}
A local measurement performed to distinguish a set of orthogonal quantum states is called an {\it orthogonality-preserving local measurement} (OPLM), if the postmeasurement states remain orthogonal.
A measurement is {\it non-trivial} if at least one POVM element is not proportional to the identity operator; otherwise, we call it a {\it trivial measurement}.
\end{definition}

Specifically, when elucidating the argument that a set of orthogonal quantum states can be distinguished via LOCC, instead of the conventional textual description, we adopt a more succinct and clear flowchart representation. The symbol $A^*$ denotes that the post-measurement states remain orthogonal on Alice's side, while the symbol $B^*$ indicates that the post-measurement states are orthogonal on Bob's side. Perfect discrimination can be achieved with only local measurement. In this work, we always stick to OPLM, and denote $P[|\cdot\rangle]_{\cal{P}}:=|\cdot\rangle\langle\cdot|_{\cal{P}}$, $P[(|i\rangle,|j\rangle,\cdots)]_{\cal{P}}:=(|i\rangle\langle i|+|j\rangle\langle j|+\cdots)_{\cal{P}}$
where ${\cal{P}}$ represents the party.

\begin{definition} {\it Locally indistinguishability} \cite{Bennett1999}. A set of orthogonal states is called locally indistinguishable if it is not possible to distinguish completely the whole set under local operations and classical communications (LOCC).
\end{definition}

\begin{definition} {\it Local irreducibility} \cite{Halder2019}. A set of orthogonal quantum states is locally irreducible if it is not possible to eliminate one or more quantum states from the set by nontrivial OPLMs.
\end{definition}

In fact, to prove the local indistinguishability (or local irreducibility), it is sufficient to verify that ``orthogonality-preserving local measurements are only trivial measurements" \cite{Walgate2002}.

\begin{definition} {\it Genuine nonlocality} \cite{Rout2019}. A set of multipartite orthogonal states is genuinely nonlocal if it is locally indistinguishable in every bipartition of the subsystems. Specifically, if a genuinely nonlocal set is locally reducible when all the parties are separated, we say it has genuine nonlocality of type \uppercase\expandafter{\romannumeral 1}; Otherwise, we say it has genuine nonlocality of type \uppercase\expandafter{\romannumeral 2}.
\end{definition}

\begin{definition} {\it local redundancy} \cite{Bandyopadhyay2021}. The original set is local redundant if it remains orthogonal after discarding one or more subsystems. Otherwise, it is termed locally irredundant.
\end{definition}

As mentioned in this reference, for a set of orthogonal states with local redundancy, the orthogonality will maintain after discarding one or more subsystems, but it may become locally distinguishable or indistinguishable. Therefore, we only consider the orthogonal sets with no redundancy in activating the nonlocality.

\begin{definition} {\it Activating nonlocality} \cite{Bandyopadhyay2021}. A set of orthogonal quantum states without local redundancy and local indistinguishability is said to be activatable if there exists an OPLM such that the postmeasurement states become locally indistinguishable.
\end{definition}

\section{Orthogonal product sets with genuine nonlocality}\label{sec:3}

To study the activation of local sets, we first consider the structure of orthogonal state sets after the implementation of~ {\textcolor{blue}{OPLM}}. Here we provide two genuinely nonlocal sets without entanglement in tripartite systems. Suppose that Alice, Bob and Charlie hold the first, second and third systems, respectively, denoted by A, B and C.

\begin{proposition}\label{prop:gn:d1d2d3}
The following  $2d_{2}+2d_{3}-4$ orthogonal product states in $\mathbb{C}^{d_{1}}\otimes\mathbb{C}^{d_{2}}\otimes\mathbb{C}^{d_{3}}$ $(4\leq d_{1}\leq d_{2}\leq d_{3})$ have genuine nonlocality:
\begin{equation*}
\begin{aligned}
|\phi_{i}\rangle=&|i\rangle_{A}|0-i\rangle_{B}|1\rangle_{C},~1\leq i \leq d_{1}-1,\\
|\phi_{i+d_{1}-1}\rangle=&|0-i\rangle_{A}|j\rangle_{B}|1\rangle_{C},~1\leq i\leq d_{1}-2,\\&j=i+1;i=d_{1}-1,j=1,\\
|\phi_{j+d_{1}-1}\rangle=&|0-1\rangle_{A}|j\rangle_{B}|1\rangle_{C},~d_{1}\leq j\leq d_{2}-1,\\
|\phi_{d_{1}+d_{2}-1}\rangle=&|(d_{1}-1)\rangle_{A}|2-d_{1}\rangle_{B}|1\rangle_{C},~d_{1}< d_{2},\\
\end{aligned}
\end{equation*}
\begin{equation}
\begin{aligned}
|\phi_{d_{2}+s_{1}-1}\rangle=&|(d_{1}-2)\rangle_{A}|(s_{1}-1)-s_{1}\rangle_{B}|1\rangle_{C},\\
|\phi_{d_{2}+t_{1}-1}\rangle=&|(d_{1}-1)\rangle_{A}|(t_{1}-1)-t_{1}\rangle_{B}|1\rangle_{C},\\
|\phi_{2d_{2}-1}\rangle=&|_{0}+_{(d_{1}-1)}\rangle_{A}|_{0}+_{(d_{2}-1)}\rangle_{B}|_{0}+_{(d_{3}-1)}\rangle_{C},\\
|\phi_{i+2d_{2}-1}\rangle=&|i\rangle_{A}|0+1\rangle_{B}|0-i\rangle_{C},~1\leq i\leq d_{1}-1,\\
|\phi_{i+d_{1}+2d_{2}-2}\rangle=&|0-i\rangle_{A}|0+1\rangle_{B}|j\rangle_{C},~1\leq i\leq d_{1}-2,\\&j=i+1,\\
|\phi_{j+d_{1}+2d_{2}-3}\rangle=&|0-1\rangle_{A} |0+1\rangle_{B}|j\rangle_{C},~d_{1}\leq j\leq d_{3}-1,\\
|\phi_{d_{1}+2d_{2}+d_{3}-3}\rangle=&|(d_{1}-1)\rangle_{A} |0+1\rangle_{B}|2-d_{1}\rangle_{C},~d_{1}< d_{3},\\
|\phi_{s_{2}+2d_{2}+d_{3}-3}\rangle=&|(d_{1}-2)\rangle_{A} |0+1\rangle_{B}|(s_{2}-1)-s_{2}\rangle_{C},\\
|\phi_{t_{2}+2d_{2}+d_{3}-3}\rangle=&|(d_{1}-1)\rangle_{A} |0+1\rangle_{B}|(t_{2}-1)-t_{2}\rangle_{C},
\end{aligned}
\end{equation}
where $s_{1}=d_{1}+2r_{1}-1$, $r_{1}=1,2,\cdots,\lfloor\frac{d_{2}-d_{1}}{2}\rfloor$, $s_{2}=d_{1}+2r_{2}-1$, $r_{2}=1,2,\cdots,\lfloor\frac{d_{3}-d_{1}}{2}\rfloor$. When $d_{2}-d_{1}$ is odd, $t_{1}=d_{1}+2r_{1}$; when $d_{2}-d_{1}$ is even, $t_{1}=d_{1}+2u_{1}$, $u_{1}=1,2,\cdots,\frac{d_{2}-d_{1}}{2}-1$. When $d_{3}-d_{1}$ is odd, $t_{2}=d_{1}+2r_{2}$; when $d_{3}-d_{1}$ is even, $t_{2}=d_{1}+2u_{2}$, $u_{2}=1,2,\cdots,\frac{d_{3}-d_{1}}{2}-1$.
\end{proposition}

The proof is given in Appendix \ref{app:prop:gn:d1d2d3}.

To activate type \uppercase\expandafter{\romannumeral 2} genuine nonlocality, we first present a locally indistinguishable set of orthogonal quantum states.

\begin{lemma}\label{lemma}
The following $2d_{2}-1$ orthogonal product states in $\mathbb{C}^{d_{1}}\otimes \mathbb{C}^{d_{2}}$ $(5\leq d_{1}+1\leq d_{2})$ cannot be perfectly distinguished under LOCC:
\begin{equation}
\begin{aligned}
 |\phi_{1}\rangle=&|(d_{1}-2)\rangle_{A}|0-1\rangle_{B},\\
 |\phi_{2}\rangle=&|(d_{1}-1)\rangle_{A}|0-2\rangle_{B},\\
 |\phi_{i}\rangle=&|(d_{1}-i)\rangle_{A}|0-i\rangle_{B},~3\leq i\leq d_{1},\\
 |\phi_{1+d_{1}}\rangle=&|(d_{1}-2)-1\rangle_{A}|2\rangle_{B},\\
 |\phi_{i+d_{1}}\rangle=&|(d_{1}-i)-(d_{1}-1)\rangle_{A}|j\rangle_{B},~2\leq i \leq d_{1}-1,\\&j=i+1;i=d_{1},j=1,\\
 |\phi_{j+d_{1}}\rangle=&|(d_{1}-2)-(d_{1}-1)\rangle_{A}|j\rangle_{B},\\&d_{1}+1\leq j\leq d_{2}-1,\\
 |\phi_{d_{1}+d_{2}}\rangle=&|0\rangle_{A}|2-(d_{1}+1)\rangle_{B},~d_{1}+1< d_{2},\\
 |\phi_{s_{1}+d_{2}-1}\rangle=&|1\rangle_{A}|(s_{1}-1)-s_{1}\rangle_{B},\\
 |\phi_{t_{1}+d_{2}-1}\rangle=&|0\rangle_{A}|(t_{1}-1)-t_{1}\rangle_{B},\\
|\phi_{2d_{2}-1}\rangle=&|_{0}+_{(d_{1}-1)}\rangle_{A}|_{0}+_{(d_{2}-1)}\rangle_{B},
\end{aligned}
\end{equation}
where $s_{1}=d_{1}+2r_{1}$, $r_{1}=1,2,\cdots,\lfloor\frac{d_{2}-d_{1}-1}{2}\rfloor$. When $d_{2}-d_{1}$ is even,  $t_{1}=d_{1}+2r_{1}+1$; when $d_{2}-d_{1}$ is odd, $t_{1}=d_{1}+2u_{1}+1$, $u_{1}=1,2,\cdots,\frac{d_{2}-d_{1}-1}{2}-1$.
\end{lemma}

The proof is given in Appendix \ref{app:lemma}. Based on the lemma, we present the second set of orthogonal product states with genuine nonlocality in $\mathbb{C}^{d_{1}}\otimes\mathbb{C}^{d_{2}}\otimes\mathbb{C}^{d_{3}}$ $(5\leq d_{1}+1\leq d_{2}\leq d_{3})$.

\begin{proposition}\label{prop:gn:d1d2d3:2}
The following $2d_{2}+2d_{3}-4$ orthogonal product states in $\mathbb{C}^{d_{1}}\otimes\mathbb{C}^{d_{2}}\otimes\mathbb{C}^{d_{3}}$ $(5\leq d_{1}+1\leq d_{2}\leq d_{3})$ have genuine nonlocality:
\begin{equation*}
\begin{aligned}
|\phi_{1}\rangle=&|(d_{1}-2)\rangle_{A}|0-1\rangle_{B}|1\rangle_{C},\\
|\phi_{2}\rangle=&|(d_{1}-1)\rangle_{A}|0-2\rangle_{B}|1\rangle_{C},\\
|\phi_{i}\rangle=&|(d_{1}-i)\rangle_{A}|0-i\rangle_{B}|1\rangle_{C},~3\leq i \leq d_{1},\\
\end{aligned}
\end{equation*}
\begin{equation}
\begin{aligned}
|\phi_{1+d_{1}}\rangle=&|(d_{1}-2)-1\rangle_{A}|2\rangle_{B}|1\rangle_{C},\\
|\phi_{i+d_{1}}\rangle=&|(d_{1}-i)-(d_{1}-1)\rangle_{A}|j\rangle_{B}|1\rangle_{C},\\&2\leq i \leq d_{1}-1,j=i+1;i=d_{1},j=1,\\
|\phi_{j+d_{1}}\rangle=&|(d_{1}-2)-(d_{1}-1)\rangle_{A}|j\rangle_{B}|1\rangle_{C},\\&d_{1}+1\leq j\leq d_{2}-1,\\
|\phi_{d_{1}+d_{2}}\rangle=&|0\rangle_{A}|2-(d_{1}+1)\rangle_{B}|1\rangle_{C},~d_{1}+1< d_{2},\\
|\phi_{s_{1}+d_{2}-1}\rangle=&|1\rangle_{A}|(s_{1}-1)-s_{1}\rangle_{B}|1\rangle_{C},\\
|\phi_{t_{1}+d_{2}-1}\rangle=&|0\rangle_{A}|(t_{1}-1)-t_{1}\rangle_{B}|1\rangle_{C},\\
|\phi_{2d_{2}-1}\rangle=&|_{0}+_{(d_{1}-1)}\rangle_{A}|_{0}+_{(d_{2}-1)}\rangle_{B}|_{0}+_{(d_{3}-1)}\rangle_{C},\\
|\phi_{2d_{2}}\rangle=&|(d_{1}-2)\rangle_{A}|0+1\rangle_{B}|0-1\rangle_{C},\\
|\phi_{2d_{2}+1}\rangle=&|(d_{1}-1)\rangle_{A}|0+1\rangle_{B}|0-2\rangle_{C},\\
|\phi_{i+2d_{2}-1}\rangle=&|(d_{1}-i)\rangle_{A}|0+1\rangle_{B}|0-i\rangle_{C},~3\leq i\leq d_{1},\\
|\phi_{d_{1}+2d_{2}}\rangle=&|(d_{1}-2)-1\rangle_{A}|0+1\rangle_{B}|2\rangle_{C},\\
|\phi_{i+d_{1}+2d_{2}-1}\rangle=&|(d_{1}-i)-(d_{1}-1)\rangle_{A}|0+1\rangle_{B}|j\rangle_{C},\\&2\leq i\leq d_{1}-1,~j=i+1,\\
|\phi_{j+d_{1}+2d_{2}-2}\rangle=&|(d_{1}-2)-(d_{1}-1)\rangle_{A} |0+1\rangle_{B}|j\rangle_{C},\\&d_{1}+1\leq j\leq d_{3}-1,\\
|\phi_{d_{1}+2d_{2}+d_{3}-2}\rangle=&|0\rangle_{A} |0+1\rangle_{B}|2-(d_{1}+1)\rangle_{C},~d_{1}+1< d_{3},\\
|\phi_{s_{2}+2d_{2}+d_{3}-3}\rangle=&|1\rangle_{A} |0+1\rangle_{B}|(s_{2}-1)-s_{2}\rangle_{C},\\
|\phi_{t_{2}+2d_{2}+d_{3}-3}\rangle=&|0\rangle_{A} |0+1\rangle_{B}|(t_{2}-1)-t_{2}\rangle_{C},
\end{aligned}
\end{equation}
where $s_{1}=d_{1}+2r_{1}$, $r_{1}=1,2,\cdots,\lfloor\frac{d_{2}-d_{1}-1}{2}\rfloor$, $s_{2}=d_{1}+2r_{2}$, $r_{2}=1,2,\cdots,\lfloor\frac{d_{3}-d_{1}-1}{2}\rfloor$. When $d_{2}-d_{1}$ is even, $t_{1}=d_{1}+2r_{1}+1$; when $d_{2}-d_{1}$ is odd, $t_{1}=d_{1}+2u_{1}+1$, $u_{1}=1,2,\cdots,\frac{d_{2}-d_{1}-1}{2}-1$. When $d_{3}-d_{1}$ is even, $t_{2}=d_{1}+2r_{2}+1$; when $d_{3}-d_{1}$ is odd, $t_{2}=d_{1}+2u_{2}+1$, $u_{2}=1,2,\cdots,\frac{d_{3}-d_{1}-1}{2}-1$.
\end{proposition}

The conclusion can be proved similar to the proof of Proposition \ref{prop:gn:d1d2d3}.

\section{Activation of type \uppercase\expandafter{\romannumeral 1}  genuine nonlocality }\label{sec:4}

In this section, we consider the activation of type-\uppercase\expandafter{\romannumeral 1} genuine nonlocality of orthogonal product states in $\mathbb{C}^{d_{1}}\otimes\mathbb{C}^{d_{2}}\otimes\mathbb{C}^{d_{3}}$ via local operations. We begin with the case of odd $d_{i}$ $(i=1,2,3)$.

\subsection{$\mathbb{C}^{d_{1}}\otimes\mathbb{C}^{d_{2}}\otimes\mathbb{C}^{d_{3}}$ ($d_{i}$ is odd)}

\begin{theorem}\label{th:1:d1d2d3:odd}
For odd integers $d_{1},d_{2}, d_{3}$, the following set of orthogonal product states in $\mathbb{C}^{d_{1}}\otimes\mathbb{C}^{d_{2}}\otimes\mathbb{C}^{d_{3}}$ ($11\leq d_{1}\leq d_{2}\leq d_{3}$) can be converted into a type-\uppercase\expandafter{\romannumeral 1} genuinely nonlocal set via {\textcolor{blue}{OPLMs by Alice}}:
\begin{equation*}
\begin{aligned}
 |\phi_{i}\rangle=&|i-(2k_{1}-i)\rangle_{A}|0-i\rangle_{B}|1\rangle_{C},\\&1\leq i\leq k_{1}-1,\\
 |\phi_{i+k_{1}-1}\rangle=&|0-i+(2k_{1}-i)-2k_{1}\rangle_{A}|j\rangle_{B}
 |1\rangle_{C},\\&1\leq i\leq k_{1}-2,j=i+1;\\& i=k_{1}-1,j=1,\\
\end{aligned}
\end{equation*}
\begin{equation*}
\begin{aligned}
 |\phi_{j+k_{1}-1}\rangle=&|0-1+(2k_{1}-1)-2k_{1}\rangle_{A}|j\rangle_{B}
 |1\rangle_{C},\\&k_{1}\leq j\leq k_{2}-1,\\
 |\phi_{k_{1}+k_{2}-1}\rangle=&|(k_{1}-1)-(k_{1}+1)\rangle_{A}|2-k_{1}\rangle_{B}
 |1\rangle_{C},\\&k_{1}<k_{2},\\
 |\phi_{s_{1}+k_{2}-1}\rangle=&|(k_{1}-2)-(k_{1}+2)\rangle_{A}
 |(s_{1}-1)-s_{1}\rangle_{B}|1\rangle_{C},\\
 |\phi_{t_{1}+k_{2}-1}\rangle=&|(k_{1}-1)-(k_{1}+1)\rangle_{A}
 |(t_{1}-1)-t_{1}\rangle_{B}|1\rangle_{C}\\
|\phi_{2k_{2}-1}\rangle=&|_{0}+_{2k_{1}}\rangle_{A}|_{0}+_{(k_{2}-1)}\rangle_{B}|_{0}+_{(k_{3}-1)}\rangle_{C},\\
 |\phi_{i+2k_{2}-1}\rangle=&|i-(2k_{1}-i)\rangle_{A}|0+1\rangle_{B}|0-i\rangle_{C},
 \\&1\leq i\leq k_{1}-1,\\
 |\phi_{i+k_{1}+2k_{2}-2}\rangle=&|0-i+(2k_{1}-i)-2k_{1}\rangle_{A}|0+1\rangle_{B}
 |j\rangle_{C},\\&1\leq i\leq k_{1}-2,j=i+1,\\
 |\phi_{j+k_{1}+2k_{2}-3}\rangle=&|0-1+(2k_{1}-1)-2k_{1}\rangle_{A}|0+1\rangle_{B}
 \\&|j\rangle_{C},~k_{1}\leq j \leq k_{3}-1,\\
 |\phi_{k_{1}+2k_{2}+k_{3}-3}\rangle=&|(k_{1}-1)-(k_{1}+1)\rangle_{A}
 |0+1\rangle_{B}\\&|2-k_{1}\rangle_{C},~k_{1}<k_{3},\\
 |\phi_{s_{2}+2k_{2}+k_{3}-3}\rangle=&|(k_{1}-2)-(k_{1}+2)\rangle_{A}
 |0+1\rangle_{B}\\&|(s_{2}-1)-s_{2}\rangle_{C},\\
 |\phi_{t_{2}+2k_{2}+k_{3}-3}\rangle=&|(k_{1}-1)-(k_{1}+1)\rangle_{A}
 |0+1\rangle_{B}\\&|(t_{2}-1)-t_{2}\rangle_{C},\\
 |\psi_{i}\rangle=&|(i+k_{1})-(k_{1}-i)\rangle_{A}|k_{2}-(i+k_{2})\rangle_{B}
 \\&|(1+k_{3})\rangle_{C},~1\leq i\leq k_{1},\\
 |\psi_{1+k_{1}}\rangle=&|k_{1}-(1+k_{1})+(k_{1}-1)-2\rangle_{A}
 \\&|(2+k_{2})\rangle_{B}|(1+k_{3})\rangle_{C},\\
 |\psi_{i+k_{1}}\rangle=&|k_{1}-(i+k_{1})+(k_{1}-i)-0\rangle_{A}\\&
 |(j+k_{2})\rangle_{B}|(1+k_{3})\rangle_{C},\\&2\leq i\leq k_{1}-2,j=i+1,\\
 |\psi_{2k_{1}-1}\rangle=&|k_{1}-(2k_{1}-1)+1-2\rangle_{A}|(k_{1}+k_{2})\rangle_{B}
 \\&|(1+k_{3})\rangle_{C},\\
 |\psi_{2k_{1}}\rangle=&|k_{1}-2k_{1}+0-1\rangle_{A}|(1+k_{2})\rangle_{B}\\&
 |(1+k_{3})\rangle_{C},\\
 |\psi_{j+k_{1}}\rangle=&|k_{1}-(1+k_{1})+(k_{1}-1)-2\rangle_{A}\\&
 |(j+k_{2})\rangle_{B}|(1+k_{3})\rangle_{C},k_{1}+1\leq j \leq k_{2},\\ |\psi_{k_{1}+k_{2}+1}\rangle=&|2k_{1}-0\rangle_{A}|(2+k_{2})-(k_{1}+k_{2}+1)
 \rangle_{B}\\&|(1+k_{3})\rangle_{C},~k_{1}<k_{2},\\
 |\psi_{s_{3}+k_{2}}\rangle=&|(2k_{1}-1)-1\rangle_{A}|(s_{3}+k_{2}-1)-(s_{3}+\\&k_{2})
 \rangle_{B}|(1+k_{3})\rangle_{C},\\ |\psi_{t_{3}+k_{2}}\rangle=&|2k_{1}-0\rangle_{A}|(t_{3}+k_{2}-1)-(t_{3}+k_{2})
 \rangle_{B}\\&|(1+k_{3})\rangle_{C},\\ |\psi_{2k_{2}+1}\rangle=&|_{0}+_{2k_{1}}\rangle_{A}|_{k_{2}}+_{2k_{2}}\rangle_{B}
 |_{k_{3}}+_{2k_{3}}\rangle_{C},\\
 |\psi_{i+2k_{2}+1}\rangle=&|(i+k_{1})-(k_{1}-i)\rangle_{A}|k_{2}+(1+k_{2})
 \rangle_{B}\\&|k_{3}-(i+k_{3})\rangle_{C},~1\leq i\leq k_{1},\\
 |\psi_{k_{1}+2k_{2}+2}\rangle=&|k_{1}-(1+k_{1})+(k_{1}-1)-2\rangle_{A}
 |k_{2}+\\&(1+k_{2})\rangle_{B}|(2+k_{3})\rangle_{C},\\
 |\psi_{i+k_{1}+2k_{2}+1}\rangle=&|k_{1}-(i+k_{1})+(k_{1}-i)-0\rangle_{A}\\&
 |k_{2}+(1+k_{2})\rangle_{B}|(j+k_{3})\rangle_{C},\\&2\leq i\leq k_{1}-2,j=i+1,\\
\end{aligned}
\end{equation*}
\begin{equation}
\begin{aligned}
 |\psi_{2k_{1}+2k_{2}}\rangle=&|k_{1}-(2k_{1}-1)+1-2\rangle_{A}|k_{2}+\\&(1+k_{2})
 \rangle_{B}|(k_{1}+k_{3})\rangle_{C},\\
 |\psi_{j+k_{1}+2k_{2}}\rangle=&|k_{1}-(1+k_{1})+(k_{1}-1)-2\rangle_{A}
 |k_{2}+\\&(1+k_{2})\rangle_{B}|(j+k_{3})\rangle_{C},~k_{1}+1\leq j \leq k_{3},\\
|\psi_{k_{1}+2k_{2}+k_{3}+1}\rangle=&|2k_{1}-0\rangle_{A}|k_{2}+(1+k_{2})
 \rangle_{B}\\&|(2+k_{3})-(k_{1}+k_{3}+1)\rangle_{C},~k_{1}<k_{3},\\
 |\psi_{s_{4}+2k_{2}+k_{3}}\rangle=&|(2k_{1}-1)-1\rangle_{A}|k_{2}+(1+k_{2})
 \rangle_{B}\\&|(s_{4}+k_{3}-1)-(s_{4}+k_{3})\rangle_{C},\\ |\psi_{t_{4}+2k_{2}+k_{3}}\rangle=&|2k_{1}-0\rangle_{A}|k_{2}+(1+k_{2})\rangle_{B}
 \\&|(t_{4}+k_{3}-1)-(t_{4}+k_{3})\rangle_{C},
\end{aligned}
\end{equation}
where $d_{i}=2k_{i}+1$ $ (i=1,2,3)$, $s_{1}=k_{1}+2r_{1}-1$, $s_{3}=k_{1}+2r_{1}$, $r_{1}=1,2,\cdots,\lfloor\frac{k_{2}-k_{1}}{2}\rfloor$, $s_{2}=k_{1}+2r_{2}-1$, $s_{4}=k_{1}+2r_{2}$, $r_{2}=1,2,\cdots,\lfloor\frac{k_{3}-k_{1}}{2}\rfloor$. When $k_{2}-k_{1}$ is odd, $t_{1}=k_{1}+2r_{1}$, $t_{3}=k_{1}+2r_{1}+1$; when $k_{2}-k_{1}$ is even, $t_{1}=k_{1}+2u_{1}$, $t_{3}=k_{1}+2u_{1}+1$, $u_{1}=1,2,\cdots,\frac{k_{2}-k_{1}}{2}-1$. When $k_{3}-k_{1}$ is odd, $t_{2}=k_{1}+2r_{2}$, $t_{4}=k_{1}+2r_{2}+1$; when $k_{3}-k_{1}$ is even, $t_{2}=k_{1}+2u_{2}$, $t_{4}=k_{1}+2u_{2}+1$, $u_{2}=1,2,\cdots,\frac{k_{3}-k_{1}}{2}-1$.
\end{theorem}

\begin{proof}
$(i)$ {\it Local distinguishability}. Define the measurements: $M_{1}^{A}=\{P_{i}^{A}=P[|i-(2k_{1}-i)\rangle]_{A}, i=1,2,\cdots,k_{1}-1, P_{k_{1}}^{A}=I-\sum_{i=1}^{k_{1}-1}P_{i}^{A}\}$, $M_{2}^{A}=\{Q_{i}^{A}=P[|(i+k_{1})-(k_{1}-i)\rangle]_{A}, i=1,2,\cdots,k_{1}, Q_{k_{1}+1}^{A}=I-\sum_{i=1}^{k_{1}}Q_{i}^{A}\}$ for Alice; $M_{1}^{B}=\{P_{i}^{B}=P[|i\rangle]_{B},i=1,\cdots,k_{2}-1, P_{k_{2}}^{B}=I-\sum_{i=1}^{k_{2}-1}P_{i}^{B}\}$, $M_{2}^{B}=\{Q_{i}^{B}=P[|(i+k_{2})\rangle]_{B},i=1,\cdots,k_{2}, Q_{k_{2}+1}^{B}=I-\sum_{i=1}^{k_{2}}Q_{i}^{B}\}$ for Bob;  $M_{1}^{C}=\{P_{i}^{C}=P[|i\rangle]_{C}, i\in \mathbb{Z}_{d_{3}}\}$ for Charlie. The details of the local discrimination protocol are shown in Fig. \ref{fig:1:odd:d1d2d3}
\begin{figure*}[h]
	\centering	\includegraphics[scale=0.5]{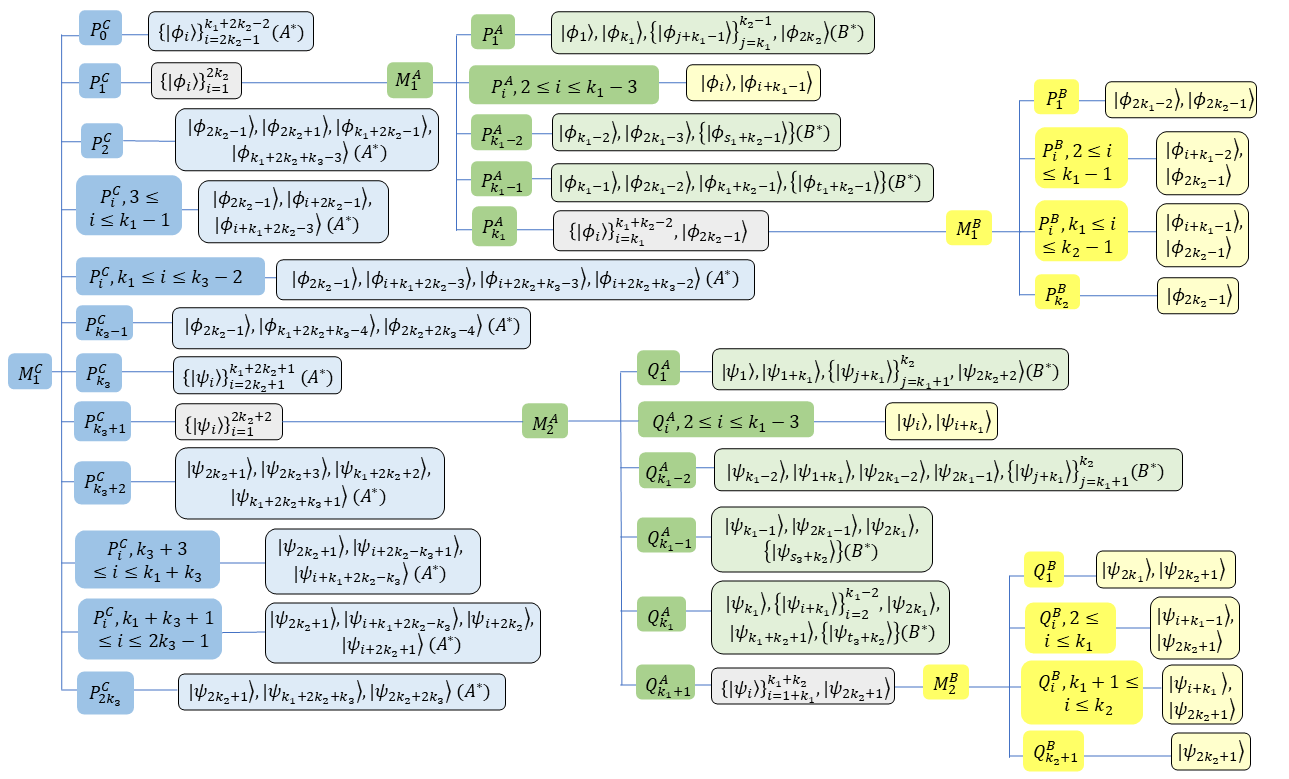}
	\caption{Distinction details for $\mathbb{C}^{d_{1}}\otimes \mathbb{C}^{d_{2}}\otimes \mathbb{C}^{d_{3}}$ ($d_{i}$ is odd)}
         \label{fig:1:odd:d1d2d3}
\end{figure*}

$(ii)$ {\it Local irredundancy}. Let $d_{i}=p_{1}^{(i)}p_{2}^{(i)}\cdots p_{n_{i}}^{(i)}=2k_{i}+1\geq11$, where $p_{l}^{(i)}$ is a prime factor of $d_{i}$ (here we assume that $3\leq p_{1}^{(i)}\leq p_{2}^{(i)}\leq \cdots \leq p_{n_{i}}^{(i)}, n_{i}\geq2, i=1,2,3$). Suppose that  $\mathcal{H}_{A}=\mathbb{C}^{d_{1}}$ can be factored into an $n_{1}$-partite system $\otimes_{l=1}^{n_{1}}\mathcal{H}_{a_{l}}$ with $\mathcal{H}_{a_{l}}=\mathbb{C}^{p_{l}^{(1)}}$; $\mathcal{H}_{B}=\mathbb{C}^{d_{2}}$ can be factored into an $n_{2}$-partite system $\otimes_{l=1}^{n_{2}}\mathcal{H}_{b_{l}}$ with $\mathcal{H}_{b_{l}}=\mathbb{C}^{p_{l}^{(2)}}$; $\mathcal{H}_{C}=\mathbb{C}^{d_{3}}$ can be factored into an $n_{3}$-partite system $\otimes_{l=1}^{n_{3}}\mathcal{H}_{c_{l}}$ with $\mathcal{H}_{c_{l}}=\mathbb{C}^{p_{l}^{(3)}}$. Denote $|\phi_{s} \rangle=|\phi_{s} \rangle_{A}|\phi_{s} \rangle_{B}|\phi_{s} \rangle_{C}$ and the discarded subsystems by $\mathcal{T}_{A} \subseteq \{a_{1},a_{2}\cdots,a_{n_{1}}\}, \mathcal{T}_{B} \subseteq \{b_{1},b_{2},\cdots,b_{n_{2}}\}$ and $\mathcal{T}_{C} \subseteq \{c_{1},c_{2},\cdots,c_{n_{3}}\}$.

For $s\neq t$, we have $\langle\phi_{s}|\phi_{t}\rangle_{B}\neq 0$ and $\langle\phi_{s}|\phi_{t}\rangle_{C}\neq 0$ when  $s,t\in\{1,\cdots,k_{1}-1\}$,
while for $s,t\in\{k_{1},\cdots,k_{1}+k_{2}-2\}$ we have $\langle\phi_{s}|\phi_{t}\rangle_{A}\neq 0$ and $\langle\phi_{s}|\phi_{t}\rangle_{C}\neq 0$, and for $s,t\in\{2k_{1}-2,k_{1}+2k_{2}-1,\cdots,k_{1}+2k_{2}+k_{3}-4\}$ we have $\langle\phi_{s}|\phi_{t}\rangle_{A}\neq 0$ and $\langle\phi_{s}|\phi_{t}\rangle_{B}\neq 0$. If the set is local redundancy, there must exist nonempty sets among $\mathcal{T}_{A}$, $\mathcal{T}_{B}$, $\mathcal{T}_{C}$ and some unitary matrices $U_{A}\in U(d_{1})$, $U_{B}\in U(d_{2})$, $U_{C}\in U(d_{3})$ such that $\{Tr_{\mathcal{T}_{A}\cup\mathcal{T}_{B}\cup\mathcal{T}_{C}}[(U_{A}\otimes U_{B}\otimes U_{C})|\phi_{s}\rangle \langle\phi_{s}|(U_{A}^{\dagger}\otimes U_{B}^{\dagger}\otimes U_{C}^{\dagger})]\}_{s=1}^{2k_{2}+2k_{3}-4}$ are pairwise orthogonal. This is equivalent to the orthogonality of the set $\{Tr_{\mathcal{T}_{A}}(U_{A}|\phi_{s}\rangle_{A}\langle\phi_{s}|U_{A}^{\dagger})
\otimes Tr_{\mathcal{T}_{B}}(U_{B}|\phi_{s}\rangle_{B}\langle\phi_{s}|U_{B}^{\dagger})
\otimes Tr_{\mathcal{T}_{C}}(U_{C}|\phi_{s}\rangle_{C}\langle\phi_{s}|U_{C}^{\dagger})
\}_{s=1}^{2k_{2}+2k_{3}-4}$. Since the partial trace operation preserves nonorthogonality, the following sets must be orthogonal to ensure the satisfaction of these conditions: $\{Tr_{\mathcal{T}_{A}}(U_{A}|\phi_{s}\rangle_{A} \langle\phi_{s}|U_{A}^{\dagger})\}_{s=1}^{k_{1}-1}$, $\{Tr_{\mathcal{T}_{B}}(U_{B}|\phi_{s}\rangle_{B} \langle\phi_{s}|U_{B}^{\dagger})\}_{s=k_{1}}^{k_{1}+k_{2}-2}$, $\{Tr_{\mathcal{T}_{C}}(U_{C}|\phi_{s}\rangle_{C} \langle\phi_{s}|$ $U_{C}^{\dagger}), s=2k_{1}-2,k_{1}+2k_{2}-1,\cdots,k_{1}+2k_{2}+k_{3}-4\}$. Without loss of generality, we assume that $\mathcal{T}_{A}\neq\emptyset$. Therefore, the resulting states $\{Tr_{\mathcal{T}_{A}}(U_{A}|\phi_{s}\rangle_{A} \langle\phi_{s}|U_{A}^{\dagger})\}_{s=1}^{k_{1}-1}$ are ($k_{1}-1$) orthogonal states in the systems corresponding to $\{a_{1},a_{2},\cdots,a_{n_{1}}\}\setminus\mathcal{T}_{A}$ whose dimension is at most $\frac{d_{1}}{3}<k_{1}-1$. Thus, we arrive at a contradiction.

$(iii)$ {\it The above set can be transformed into a genuinely nonlocal set.} Suppose that Alice performs the measurement $K^{A}=\{K_{1}^{A}, K_{2}^{A}\}$, where $K_{1}^{A}=\sum_{j=0}^{k_{1}-1}|j\rangle_{A} \langle j|$ and $K_{2}^{A}=\sum_{j=k_{1}}^{2k_{1}}|j\rangle_{A} \langle j|$.

If the outcome of $K^{A}$ is `1', the states are transferred as follows:
\begin{equation*}
\begin{aligned}
 |\phi_{i}\rangle=&|i\rangle_{A}|0-i\rangle_{B}|1\rangle_{C},~1\leq i\leq k_{1}-1,\\
 |\phi_{i+k_{1}-1}\rangle=&|0-i\rangle_{A}|j\rangle_{B}
 |1\rangle_{C},~1\leq i\leq k_{1}-2,\\&j=i+1;i=k_{1}-1,j=1,\\
 |\phi_{j+k_{1}-1}\rangle=&|0-1\rangle_{A}|j\rangle_{B}
 |1\rangle_{C},~k_{1}\leq j\leq k_{2}-1,\\
 |\phi_{k_{1}+k_{2}-1}\rangle=&|(k_{1}-1)\rangle_{A}|2-k_{1}\rangle_{B}
 |1\rangle_{C},~k_{1}<k_{2},\\
 |\phi_{s_{1}+k_{2}-1}\rangle=&|(k_{1}-2)\rangle_{A}
 |(s_{1}-1)-s_{1}\rangle_{B}|1\rangle_{C},\\
 |\phi_{t_{1}+k_{2}-1}\rangle=&|(k_{1}-1)\rangle_{A}
 |(t_{1}-1)-t_{1}\rangle_{B}|1\rangle_{C},\\
|\phi_{2k_{2}-1}\rangle=&|_{0}+_{(k_{1}-1)}\rangle_{A}|_{0}+_{(k_{2}-1)}\rangle_{B}
 |_{0}+_{(k_{3}-1)}\rangle_{C},\\ |\phi_{i+2k_{2}-1}\rangle=&|i\rangle_{A}|0+1\rangle_{B}|0-i\rangle_{C},
 ~1\leq i\leq k_{1}-1,\\
 |\phi_{i+k_{1}+2k_{2}-2}\rangle=&|0-i\rangle_{A}|0+1\rangle_{B}
 |j\rangle_{C},~1\leq i\leq k_{1}-2,\\&j=i+1,\\
 |\phi_{j+k_{1}+2k_{2}-3}\rangle=&|0-1\rangle_{A}|0+1\rangle_{B}
 |j\rangle_{C},~k_{1}\leq j \leq k_{3}-1,\\
 |\phi_{k_{1}+2k_{2}+k_{3}-3}\rangle=&|(k_{1}-1)\rangle_{A}
 |0+1\rangle_{B}|2-k_{1}\rangle_{C},~k_{1}<k_{3},\\
 |\phi_{s_{2}+2k_{2}+k_{3}-3}\rangle=&|(k_{1}-2)\rangle_{A}
 |0+1\rangle_{B}|(s_{2}-1)-s_{2}\rangle_{C},\\
 |\phi_{t_{2}+2k_{2}+k_{3}-3}\rangle=&|(k_{1}-1)\rangle_{A}
 |0+1\rangle_{B}|(t_{2}-1)-t_{2}\rangle_{C},\\
 |\psi_{i}\rangle=&|(k_{1}-i)\rangle_{A}|k_{2}-(i+k_{2})\rangle_{B}
 |(1+k_{3})\rangle_{C},\\&1\leq i\leq k_{1},\\
 |\psi_{1+k_{1}}\rangle=&|(k_{1}-1)-2\rangle_{A}
 |(2+k_{2})\rangle_{B}|(1+k_{3})\rangle_{C},\\
 |\psi_{i+k_{1}}\rangle=&|(k_{1}-i)-0\rangle_{A}
 |(j+k_{2})\rangle_{B}|(1+k_{3})\rangle_{C},\\&2\leq i\leq k_{1}-2,j=i+1,\\
 |\psi_{2k_{1}-1}\rangle=&|1-2\rangle_{A}|(k_{1}+k_{2})\rangle_{B}
 |(1+k_{3})\rangle_{C},\\
 |\psi_{2k_{1}}\rangle=&|0-1\rangle_{A}|(1+k_{2})\rangle_{B}
 |(1+k_{3})\rangle_{C},\\
 |\psi_{j+k_{1}}\rangle=&|(k_{1}-1)-2\rangle_{A}
 |(j+k_{2})\rangle_{B}|(1+k_{3})\rangle_{C},\\&k_{1}+1\leq j \leq k_{2},\\
 |\psi_{k_{1}+k_{2}+1}\rangle=&|0\rangle_{A}|(2+k_{2})-(k_{1}+k_{2}+1)
 \rangle_{B}\\&|(1+k_{3})\rangle_{C},~k_{1}<k_{2},\\
 |\psi_{s_{3}+k_{2}}\rangle=&|1\rangle_{A}|(s_{3}+k_{2}-1)-(s_{3}+k_{2})
 \rangle_{B}|(1+k_{3})\rangle_{C},\\
 |\psi_{t_{3}+k_{2}}\rangle=&|0\rangle_{A}|(t_{3}+k_{2}-1)-(t_{3}+k_{2})
 \rangle_{B}|(1+k_{3})\rangle_{C},\\
|\psi_{2k_{2}+1}\rangle=&|_{0}+_{(k_{1}-1)}\rangle_{A}|_{k_{2}}+_{2k_{2}}\rangle_{B}
 |_{k_{3}}+_{2k_{3}}\rangle_{C},\\
 |\psi_{i+2k_{2}+1}\rangle=&|(k_{1}-i)\rangle_{A}|k_{2}+(1+k_{2})
 \rangle_{B}|k_{3}-\\&(i+k_{3})\rangle_{C},~1\leq i\leq k_{1},\\
 |\psi_{k_{1}+2k_{2}+2}\rangle=&|(k_{1}-1)-2\rangle_{A}
 |k_{2}+(1+k_{2})\rangle_{B}\\&|(2+k_{3})\rangle_{C},\\
\end{aligned}
\end{equation*}
\begin{equation}\label{eq:1:d1d2d3:1}
\begin{aligned}
 |\psi_{i+k_{1}+2k_{2}+1}\rangle=&|(k_{1}-i)-0\rangle_{A}
 |k_{2}+(1+k_{2})\rangle_{B}\\&|(j+k_{3})\rangle_{C},~2\leq i\leq k_{1}-2,j=i+1,\\
 |\psi_{2k_{1}+2k_{2}}\rangle=&|1-2\rangle_{A}|k_{2}+(1+k_{2})
 \rangle_{B}|(k_{1}+k_{3})\rangle_{C},\\
 |\psi_{j+k_{1}+2k_{2}}\rangle=&|(k_{1}-1)-2\rangle_{A}
 |k_{2}+(1+k_{2})\rangle_{B}\\&|(j+k_{3})\rangle_{C},~k_{1}+1\leq j \leq k_{3},\\
|\psi_{k_{1}+2k_{2}+k_{3}+1}\rangle=&|0\rangle_{A}|k_{2}+(1+k_{2})
 \rangle_{B}|(2+k_{3})-\\&(k_{1}+k_{3}+1)\rangle_{C},~k_{1}<k_{3},\\
|\psi_{s_{4}+2k_{2}+k_{3}}\rangle=&|1\rangle_{A}|k_{2}+(1+k_{2})
 \rangle_{B}|(s_{4}+k_{3}-1)-\\&(s_{4}+k_{3})\rangle_{C},\\
|\psi_{t_{4}+2k_{2}+k_{3}}\rangle=&|0\rangle_{A}|k_{2}+(1+k_{2})\rangle_{B}
 |(t_{4}+k_{3}-1)-\\&(t_{4}+k_{3})\rangle_{C},\\
\end{aligned}
\end{equation}
where $d_{i}=2k_{i}+1(i=1,2,3)$, $s_{1}=k_{1}+2r_{1}-1$, $s_{3}=k_{1}+2r_{1}$, $r_{1}=1,2,\cdots,\lfloor\frac{k_{2}-k_{1}}{2}\rfloor$, $s_{2}=k_{1}+2r_{2}-1$, $s_{4}=k_{1}+2r_{2}$, $r_{2}=1,2,\cdots,\lfloor\frac{k_{3}-k_{1}}{2}\rfloor$. When $k_{2}-k_{1}$ is odd, $t_{1}=k_{1}+2r_{1}$, $t_{3}=k_{1}+2r_{1}+1$; when $k_{2}-k_{1}$ is even, $t_{1}=k_{1}+2u_{1}$, $t_{3}=k_{1}+2u_{1}+1$, $u_{1}=1,2,\cdots,\frac{k_{2}-k_{1}}{2}-1$. When $k_{3}-k_{1}$ is odd, $t_{2}=k_{1}+2r_{2}$, $t_{4}=k_{1}+2r_{2}+1$; when $k_{3}-k_{1}$ is even, $t_{2}=k_{1}+2u_{2}$, $t_{4}=k_{1}+2u_{2}+1$, $u_{2}=1,2,\cdots,\frac{k_{3}-k_{1}}{2}-1$.

From Proposition \ref{prop:gn:d1d2d3}, it can be seen that $\{|\phi_{s}\rangle\}_{s=1}^{2k_{2}+2k_{3}-4}$ is a genuinely nonlocal set in $\mathbb{C}^{k_{1}}\otimes \mathbb{C}^{k_{2}}\otimes \mathbb{C}^{k_{3}}$. When the parties are separated, Charlie performs the measurement $\{K_{1}^{C}=\sum_{j=0}^{k_{3}-1}|j\rangle_{C} \langle j|, K_{2}^{C}=\sum_{j=k_{3}}^{2k_{3}}|j\rangle_{C} \langle j|\}$ to convert the above set to $\{|\phi_{s}\rangle\}_{s=1}^{2k_{2}+2k_{3}-4}$ and $\{|\psi_{s}\rangle\}_{s=1}^{2k_{2}+2k_{3}}$. Hence the orthogonal product states in Eq.(\ref{eq:1:d1d2d3:1}) has genuine nonlocality of type \uppercase\expandafter{\romannumeral 1}.

If the outcome of $K^{A}$ is `2', the states are transferred as follows:
\begin{equation*}
\begin{aligned}
 |\phi_{i}\rangle=&|(2k_{1}-i)\rangle_{A}|0-i\rangle_{B}|1\rangle_{C},~1\leq i\leq k_{1}-1,\\
 |\phi_{i+k_{1}-1}\rangle=&|(2k_{1}-i)-2k_{1}\rangle_{A}|j\rangle_{B}
 |1\rangle_{C},\\&1\leq i\leq k_{1}-2,j=i+1;\\&i=k_{1}-1,j=1,\\
 |\phi_{j+k_{1}-1}\rangle=&|(2k_{1}-1)-2k_{1}\rangle_{A}|j\rangle_{B}
 |1\rangle_{C},\\&k_{1}\leq j\leq k_{2}-1,\\
 |\phi_{k_{1}+k_{2}-1}\rangle=&|(k_{1}+1)\rangle_{A}|2-k_{1}\rangle_{B}
 |1\rangle_{C},~k_{1}<k_{2},\\
 |\phi_{s_{1}+k_{2}-1}\rangle=&|(k_{1}+2)\rangle_{A}
 |(s_{1}-1)-s_{1}\rangle_{B}|1\rangle_{C},\\
 |\phi_{t_{1}+k_{2}-1}\rangle=&|(k_{1}+1)\rangle_{A}
 |(t_{1}-1)-t_{1}\rangle_{B}|1\rangle_{C},\\
 |\phi_{2k_{2}-1}\rangle=&|_{k_{1}}+_{2k_{1}}\rangle_{A}|_{0}+_{(k_{2}-1)}\rangle_{B}
 |_{0}+_{(k_{3}-1)}\rangle_{C},\\
 |\phi_{i+2k_{2}-1}\rangle=&|(2k_{1}-i)\rangle_{A}|0+1\rangle_{B}|0-i\rangle_{C},
 \\&1\leq i\leq k_{1}-1,\\
 |\phi_{i+k_{1}+2k_{2}-2}\rangle=&|(2k_{1}-i)-2k_{1}\rangle_{A}|0+1\rangle_{B}
 |j\rangle_{C},\\&1\leq i\leq k_{1}-2,j=i+1,\\
 |\phi_{j+k_{1}+2k_{2}-3}\rangle=&|(2k_{1}-1)-2k_{1}\rangle_{A}|0+1\rangle_{B}
 |j\rangle_{C},\\&k_{1}\leq j \leq k_{3}-1,\\
 |\phi_{k_{1}+2k_{2}+k_{3}-3}\rangle=&|(k_{1}+1)\rangle_{A}
 |0+1\rangle_{B}|2-k_{1}\rangle_{C},~k_{1}<k_{3},\\
 |\phi_{s_{2}+2k_{2}+k_{3}-3}\rangle=&|(k_{1}+2)\rangle_{A}
 |0+1\rangle_{B}|(s_{2}-1)-s_{2}\rangle_{C},\\
 |\phi_{t_{2}+2k_{2}+k_{3}-3}\rangle=&|(k_{1}+1)\rangle_{A}
 |0+1\rangle_{B}|(t_{2}-1)-t_{2}\rangle_{C},\\
 \end{aligned}
\end{equation*}
\begin{equation}\label{eq:1:d1d2d3:2}
 \begin{aligned}
 |\psi_{i}\rangle=&|(i+k_{1})\rangle_{A}|k_{2}-(i+k_{2})\rangle_{B}
 |(1+k_{3})\rangle_{C},\\&1\leq i\leq k_{1},\\
 |\psi_{1+k_{1}}\rangle=&|k_{1}-(1+k_{1})\rangle_{A}
 |(2+k_{2})\rangle_{B}|(1+k_{3})\rangle_{C},\\
 |\psi_{i+k_{1}}\rangle=&|k_{1}-(i+k_{1})\rangle_{A}
 |(j+k_{2})\rangle_{B}|(1+k_{3})\rangle_{C},\\&2\leq i\leq k_{1}-2,j=i+1,\\
 |\psi_{2k_{1}-1}\rangle=&|k_{1}-(2k_{1}-1)\rangle_{A}|(k_{1}+k_{2})\rangle_{B}\\&
 |(1+k_{3})\rangle_{C},\\
|\psi_{2k_{1}}\rangle=&|k_{1}-2k_{1}\rangle_{A}|(1+k_{2})\rangle_{B}
 |(1+k_{3})\rangle_{C},\\
 |\psi_{j+k_{1}}\rangle=&|k_{1}-(1+k_{1})\rangle_{A}
 |(j+k_{2})\rangle_{B}|(1+k_{3})\rangle_{C},\\&k_{1}+1\leq j \leq k_{2},\\
|\psi_{k_{1}+k_{2}+1}\rangle=&|2k_{1}\rangle_{A}|(2+k_{2})-(k_{1}+k_{2}+1)
 \rangle_{B}\\&|(1+k_{3})\rangle_{C},~k_{1}<k_{2},\\
 |\psi_{s_{3}+k_{2}}\rangle=&|(2k_{1}-1)\rangle_{A}|(s_{3}+k_{2}-1)-(s_{3}+k_{2})
 \rangle_{B}\\&|(1+k_{3})\rangle_{C},\\
 |\psi_{t_{3}+k_{2}}\rangle=&|2k_{1}\rangle_{A}|(t_{3}+k_{2}-1)-(t_{3}+k_{2})
 \rangle_{B}\\&|(1+k_{3})\rangle_{C},\\
|\psi_{2k_{2}+1}\rangle=&|_{k_{1}}+_{2k_{1}}\rangle_{A}|_{k_{2}}+_{2k_{2}}\rangle_{B}
 |_{k_{3}}+_{2k_{3}}\rangle_{C},\\
 |\psi_{i+2k_{2}+1}\rangle=&|(i+k_{1})\rangle_{A}|k_{2}+(1+k_{2})
 \rangle_{B}|k_{3}-\\&(i+k_{3})\rangle_{C},~1\leq i\leq k_{1},\\
 |\psi_{k_{1}+2k_{2}+2}\rangle=&|k_{1}-(1+k_{1})\rangle_{A}
 |k_{2}+(1+k_{2})\rangle_{B}\\&|(2+k_{3})\rangle_{C},\\
 |\psi_{i+k_{1}+2k_{2}+1}\rangle=&|k_{1}-(i+k_{1})\rangle_{A}
 |k_{2}+(1+k_{2})\rangle_{B}\\&|(j+k_{3})\rangle_{C},~2\leq i\leq k_{1}-2,j=i+1,\\
 |\psi_{2k_{1}+2k_{2}}\rangle=&|k_{1}-(2k_{1}-1)\rangle_{A}|k_{2}+(1+k_{2})
 \rangle_{B}\\&|(k_{1}+k_{3})\rangle_{C},\\
 |\psi_{j+k_{1}+2k_{2}}\rangle=&|k_{1}-(1+k_{1})\rangle_{A}
 |k_{2}+(1+k_{2})\rangle_{B}\\&|(j+k_{3})\rangle_{C},~k_{1}+1\leq j \leq k_{3},\\ |\psi_{k_{1}+2k_{2}+k_{3}+1}\rangle=&|2k_{1}\rangle_{A}|k_{2}+(1+k_{2})
 \rangle_{B}|(2+k_{3})\\&-(k_{1}+k_{3}+1)\rangle_{C},~k_{1}<k_{3},\\
 |\psi_{s_{4}+2k_{2}+k_{3}}\rangle=&|(2k_{1}-1)\rangle_{A}|k_{2}+(1+k_{2})
 \rangle_{B}|(s_{4}+k_{3}\\&-1)-(s_{4}+k_{3})\rangle_{C},\\ |\psi_{t_{4}+2k_{2}+k_{3}}\rangle=&|2k_{1}\rangle_{A}|k_{2}+(1+k_{2})\rangle_{B}
 |(t_{4}+k_{3}-1)-\\&(t_{4}+k_{3})\rangle_{C},\\
\end{aligned}
\end{equation}
where $d_{i}=2k_{i}+1(i=1,2,3)$, $s_{1}=k_{1}+2r_{1}-1$, $s_{3}=k_{1}+2r_{1}$, $r_{1}=1,2,\cdots,\lfloor\frac{k_{2}-k_{1}}{2}\rfloor$, $s_{2}=k_{1}+2r_{2}-1$, $s_{4}=k_{1}+2r_{2}$, $r_{2}=1,2,\cdots,\lfloor\frac{k_{3}-k_{1}}{2}\rfloor$. When $k_{2}-k_{1}$ is odd, $t_{1}=k_{1}+2r_{1}$, $t_{3}=k_{1}+2r_{1}+1$; when $k_{2}-k_{1}$ is even, $t_{1}=k_{1}+2u_{1}$, $t_{3}=k_{1}+2u_{1}+1$, $u_{1}=1,2,\cdots,\frac{k_{2}-k_{1}}{2}-1$. When $k_{3}-k_{1}$ is odd, $t_{2}=k_{1}+2r_{2}$, $t_{4}=k_{1}+2r_{2}+1$; when $k_{3}-k_{1}$ is even, $t_{2}=k_{1}+2u_{2}$, $t_{4}=k_{1}+2u_{2}+1$, $u_{2}=1,2,\cdots,\frac{k_{3}-k_{1}}{2}-1$.

From Proposition \ref{prop:gn:d1d2d3}, it can be seen that $\{|\psi_{s}\rangle\}_{s=1}^{2k_{2}+2k_{3}}$ is a genuinely nonlocal set in $\mathbb{C}^{k_{1}+1}\otimes \mathbb{C}^{k_{2}+1}\otimes \mathbb{C}^{k_{3}+1}$. When the parties are separated, Charlie performs the measurement $\{K_{1}^{C}=\sum_{j=0}^{k_{3}-1}|j\rangle_{C} \langle j|, K_{2}^{C}=\sum_{j=k_{3}}^{2k_{3}}|j\rangle_{C} \langle j|\}$ to convert the above set to $\{|\phi_{s}\rangle\}_{s=1}^{2k_{2}+2k_{3}-4}$ and $\{|\psi_{s}\rangle\}_{s=1}^{2k_{2}+2k_{3}}$. Hence the orthogonal product states in Eq.(\ref{eq:1:d1d2d3:2}) is locally reducible, so it has genuine nonlocality of type \uppercase\expandafter{\romannumeral 1}.
\end{proof}

\subsection{$\mathbb{C}^{d_{1}}\otimes\mathbb{C}^{d_{2}}\otimes\mathbb{C}^{d_{3}} (d_{i}$ is even)}

\begin{theorem}\label{th:1:d1d2d3:even}
Let $d_{i}$ $(i=1,2,3)$ be even integers. The following orthogonal set without entanglement in $\mathbb{C}^{d_{1}}\otimes\mathbb{C}^{d_{2}}\otimes\mathbb{C}^{d_{3}}$ ($8\leq d_{1}\leq d_{2}\leq d_{3}$) can be deterministically converted to a type-\uppercase\expandafter{\romannumeral 1} genuinely nonlocal set via {\textcolor{blue}{OPLMs by Alice}}:
\begin{equation*}
\begin{aligned}
 |\phi_{i}\rangle=&|i-(2k_{1}-1-i)\rangle_{A}|0-i\rangle_{B}|1\rangle_{C},\\&1\leq i\leq k_{1}-1,\\
 |\phi_{i+k_{1}-1}\rangle=&|0-i+(2k_{1}-1-i)-(2k_{1}-1)\rangle_{A}|j\rangle_{B}\\&
 |1\rangle_{C},~1\leq i\leq k_{1}-2,j=i+1;\\&i=k_{1}-1,j=1,\\
 |\phi_{j+k_{1}-1}\rangle=&|0-1+(2k_{1}-2)-(2k_{1}-1)\rangle_{A}|j\rangle_{B}\\&
 |1\rangle_{C},~k_{1}\leq j\leq k_{2}-1,\\
 |\phi_{k_{1}+k_{2}-1}\rangle=&|(k_{1}-1)-k_{1}\rangle_{A}|2-k_{1}\rangle_{B}
 |1\rangle_{C},~k_{1}<k_{2},\\
 |\phi_{s_{1}+k_{2}-1}\rangle=&|(k_{1}-2)-(k_{1}+1)\rangle_{A}
 |(s_{1}-1)-s_{1}\rangle_{B}\\&|1\rangle_{C},\\
 |\phi_{t_{1}+k_{2}-1}\rangle=&|(k_{1}-1)-k_{1}\rangle_{A}
 |(t_{1}-1)-t_{1}\rangle_{B}|1\rangle_{C},\\ |\phi_{2k_{2}-1}\rangle=&|_{0}+_{(2k_{1}-1)}\rangle_{A}|_{0}+_{(k_{2}-1)}\rangle_{B}
 |_{0}+_{(k_{3}-1)}\rangle_{C},\\
 |\phi_{i+2k_{2}-1}\rangle=&|i-(2k_{1}-1-i)\rangle_{A}|0+1\rangle_{B}|0-i\rangle_{C},
 \\&1\leq i\leq k_{1}-1,\\
  |\phi_{i+k_{1}+2k_{2}-2}\rangle=&|0-i+(2k_{1}-1-i)-(2k_{1}-1)\rangle_{A}\\&|0+1\rangle_{B}
 |j\rangle_{C},~1\leq i\leq k_{1}-2,j=i+1,\\
 |\phi_{j+k_{1}+2k_{2}-3}\rangle=&|0-1+(2k_{1}-2)-(2k_{1}-1)\rangle_{A}|0+1\rangle_{B}\\&
 |j\rangle_{C},~k_{1}\leq j \leq k_{3}-1,\\
 |\phi_{k_{1}+2k_{2}+k_{3}-3}\rangle=&|(k_{1}-1)-k_{1}\rangle_{A}
 |0+1\rangle_{B}|2-k_{1}\rangle_{C},\\&k_{1}<k_{3},\\
 |\phi_{s_{2}+2k_{2}+k_{3}-3}\rangle=&|(k_{1}-2)-(k_{1}+1)\rangle_{A}
 |0+1\rangle_{B}\\&|(s_{2}-1)-s_{2}\rangle_{C},\\
 |\phi_{t_{2}+2k_{2}+k_{3}-3}\rangle=&|(k_{1}-1)-k_{1}\rangle_{A}
 |0+1\rangle_{B}|(t_{2}-1)-t_{2}\rangle_{C},\\
 |\phi_{2k_{2}+2k_{3}-3}\rangle=&|0-1+(2k_{1}-2)-(2k_{1}-1)\rangle_{A}|0+1\rangle_{B}\\&
 |k_{3}\rangle_{C},\\
 |\phi_{2k_{2}+2k_{3}-2}\rangle=&|0-1+(2k_{1}-2)-(2k_{1}-1)\rangle_{A}|0+1\rangle_{B}\\&
 |(k_{3}+1)\rangle_{C},\\
  |\phi_{2k_{2}+2k_{3}-1}\rangle=&|0-1+(2k_{1}-2)-(2k_{1}-1)\rangle_{A}|k_{2}\rangle_{B}
 |1\rangle_{C},\\
 |\phi_{2k_{2}+2k_{3}}\rangle=&|0-1+(2k_{1}-2)-(2k_{1}-1)\rangle_{A}\\&|(k_{2}+1)\rangle_{B}
 |1\rangle_{C},\\
 |\psi_{i}\rangle=&|(i+k_{1})-(k_{1}-1-i)\rangle_{A}|k_{2}-(i+k_{2})\rangle_{B}\\&
 |(1+k_{3})\rangle_{C},~1\leq i\leq k_{1}-1,\\
 |\psi_{i+k_{1}-1}\rangle=&|k_{1}-(i+k_{1})+(k_{1}-1-i)-(k_{1}-1)\rangle_{A}\\&
 |(j+k_{2})\rangle_{B}|(1+k_{3})\rangle_{C},1\leq i\leq k_{1}-2,\\&j=i+1;i=k_{1}-1,j=1,\\
   |\psi_{j+k_{1}-1}\rangle=&|k_{1}-(1+k_{1})+(k_{1}-2)-(k_{1}-1)\rangle_{A}\\&
 |(j+k_{2})\rangle_{B}|(1+k_{3})\rangle_{C},k_{1}\leq j \leq k_{2}-1,\\
 \end{aligned}
\end{equation*}

\begin{equation}
 \begin{aligned}
 |\psi_{k_{1}+k_{2}-1}\rangle=&|(2k_{1}-1)-0\rangle_{A}|(2+k_{2})-(k_{1}+k_{2})
 \rangle_{B}\\&|(1+k_{3})\rangle_{C},~k_{1}<k_{2},\\
  |\psi_{s_{1}+k_{2}-1}\rangle=&|(2k_{1}-2)-1\rangle_{A}|(s_{1}-1+k_{2})-(s_{1}+k_{2})
 \rangle_{B}\\&|(1+k_{3})\rangle_{C},\\
 |\psi_{t_{1}+k_{2}-1}\rangle=&|(2k_{1}-1)-0\rangle_{A}|(t_{1}-1+k_{2})-(t_{1}+k_{2})
 \rangle_{B}\\&|(1+k_{3})\rangle_{C},\\
  |\psi_{2k_{2}-1}\rangle=&|_{0}+_{(2k_{1}-1)}\rangle_{A}|_{k_{2}}+_{(2k_{2}-1)}\rangle_{B}
 |_{k_{3}}+_{(2k_{3}-1)}\rangle_{C},\\
 |\psi_{i+2k_{2}-1}\rangle=&|(i+k_{1})-(k_{1}-1-i)\rangle_{A}|k_{2}+(1+k_{2})
 \rangle_{B}\\&|k_{3}-(i+k_{3})\rangle_{C},~1\leq i\leq k_{1}-1,\\
 |\psi_{i+k_{1}+2k_{2}-2}\rangle=&|k_{1}-(i+k_{1})+(k_{1}-1-i)-(k_{1}-1)\rangle_{A}\\&
 |k_{2}+(1+k_{2})\rangle_{B}|(j+k_{3})\rangle_{C},\\&1\leq i\leq k_{1}-2,j=i+1,\\
 |\psi_{j+k_{1}+2k_{2}-3}\rangle=&|k_{1}-(1+k_{1})+(k_{1}-2)-(k_{1}-1)\rangle_{A}\\&
 |k_{2}+(1+k_{2})\rangle_{B}|(j+k_{3})\rangle_{C},\\&k_{1}\leq j \leq k_{3}-1,\\
 |\psi_{k_{1}+2k_{2}+k_{3}-3}\rangle=&|(2k_{1}-1)-0\rangle_{A}|k_{2}+(1+k_{2})
 \rangle_{B}\\&|(2+k_{3})-(k_{1}+k_{3})\rangle_{C},~k_{1}<k_{3},\\
 |\psi_{s_{2}+2k_{2}+k_{3}-3}\rangle=&|(2k_{1}-2)-1\rangle_{A}|k_{2}+(1+k_{2})
 \rangle_{B}\\&|(s_{2}+k_{3}-1)-(s_{2}+k_{3})\rangle_{C},\\
 |\psi_{t_{2}+2k_{2}+k_{3}-3}\rangle=&|(2k_{1}-1)-0\rangle_{A}|k_{2}+(1+k_{2})\rangle_{B}\\&
 |(t_{2}+k_{3}-1)-(t_{2}+k_{3})\rangle_{C},\\
\end{aligned}
\end{equation}
where $d_{i}=2k_{i}(i=1,2,3)$, $s_{1}=k_{1}+2r_{1}-1$, $r_{1}=1,2,\cdots,\lfloor\frac{k_{2}-k_{1}}{2}\rfloor$, $s_{2}=k_{1}+2r_{2}-1$, $r_{2}=1,2,\cdots,\lfloor\frac{k_{3}-k_{1}}{2}\rfloor$. When $k_{2}-k_{1}$ is odd, $t_{1}=k_{1}+2r_{1}$; when $k_{2}-k_{1}$ is even, $t_{1}=k_{1}+2u_{1}$, $u_{1}=1,2,\cdots,\frac{k_{2}-k_{1}}{2}-1$. When $k_{3}-k_{1}$ is odd, $t_{2}=k_{1}+2r_{2}$; when $k_{3}-k_{1}$ is even, $t_{2}=k_{1}+2u_{2}$, $u_{2}=1,2,\cdots,\frac{k_{3}-k_{1}}{2}-1$.
\end{theorem}

By observing the characteristics of the above constructions, it is not difficult to derive a similar structure in general tripartite system $\mathbb{C}^{d_{1}}\otimes\mathbb{C}^{d_{2}}\otimes\mathbb{C}^{d_{3}}$ ($d_{1}\leq d_{2}\leq d_{3}$, where $d_{i}$ is even more than 8 or odd more than 11).

\section{{Activation of type \uppercase\expandafter{\romannumeral 2}  genuine nonlocality}}\label{sec:5}

In this section, we activate type-\uppercase\expandafter{\romannumeral 2} genuine nonlocality of orthogonal product states in $\mathbb{C}^{d_{1}}\otimes\mathbb{C}^{d_{2}}\otimes\mathbb{C}^{d_{3}}$ .

\subsection{$\mathbb{C}^{d_{1}}\otimes\mathbb{C}^{d_{2}}\otimes\mathbb{C}^{d_{3}} (d_{1}$ is even)}

\begin{theorem}\label{th:2:d1d2d3:even}
The type-\uppercase\expandafter{\romannumeral 2} genuine nonlocality of the following $2d_{2}+2d_{3}-4$ orthogonal product states in $\mathbb{C}^{d_{1}}\otimes\mathbb{C}^{d_{2}}\otimes\mathbb{C}^{d_{3}}$ ($4\leq \frac{d_{1}}{2}\leq d_{2}\leq d_{3}$, $d_{1}$ is even) can be activated via~{\textcolor{blue}{OPLMs by Alice}}:
\begin{equation}
\begin{aligned}
 |\phi_{i}\rangle=&|i-(2k-1-i)\rangle_{A}|0-i\rangle_{B}|1\rangle_{C},\\&1\leq i\leq k-1,\\
 |\phi_{i+k-1}\rangle=&|0-i+(2k-1-i)-(2k-1)\rangle_{A}|j\rangle_{B}\\&
 |1\rangle_{C},~1\leq i\leq k-2,j=i+1;i=k-1,\\&j=1,\\
 |\phi_{j+k-1}\rangle=&|0-1+(2k-2)-(2k-1)\rangle_{A}|j\rangle_{B}
 |1\rangle_{C},\\&k\leq j\leq d_{2}-1,\\
 |\phi_{k+d_{2}-1}\rangle=&|(k-1)-k\rangle_{A}|2-k\rangle_{B}
 |1\rangle_{C},~k<d_{2},\\
 |\phi_{s_{1}+d_{2}-1}\rangle=&|(k-2)-(k+1)\rangle_{A}
 |(s_{1}-1)-s_{1}\rangle_{B}|1\rangle_{C},\\
 |\phi_{t_{1}+d_{2}-1}\rangle=&|(k-1)-k\rangle_{A}
 |(t_{1}-1)-t_{1}\rangle_{B}|1\rangle_{C},\\ |\phi_{2d_{2}-1}\rangle=&|_{0}+_{(d_{1}-1)}\rangle_{A}|_{0}+_{(d_{2}-1)}\rangle_{B}
 |_{0}+_{(d_{3}-1)}\rangle_{C},\\
 |\phi_{i+2d_{2}-1}\rangle=&|i-(2k-1-i)\rangle_{A}|0+1\rangle_{B}|0-i\rangle_{C},
 \\&1\leq i\leq k-1,\\
 |\phi_{i+k+2d_{2}-2}\rangle=&|0-i+(2k-1-i)-(2k-1)\rangle_{A}\\&|0+1\rangle_{B}
 |j\rangle_{C},~1\leq i\leq k-2,j=i+1,\\
 |\phi_{j+k+2d_{2}-3}\rangle=&|0-1+(2k-2)-(2k-1)\rangle_{A}|0+1\rangle_{B}\\&
 |j\rangle_{C},~k\leq j \leq d_{3}-1,\\
 |\phi_{k+2d_{2}+d_{3}-3}\rangle=&|(k-1)-k\rangle_{A}
 |0+1\rangle_{B}|2-k\rangle_{C},~k<d_{3},\\
 |\phi_{s_{2}+2d_{2}+d_{3}-3}\rangle=&|(k-2)-(k+1)\rangle_{A}
 |0+1\rangle_{B}\\&|(s_{2}-1)-s_{2}\rangle_{C},\\
 |\phi_{t_{2}+2d_{2}+d_{3}-3}\rangle=&|(k-1)-k\rangle_{A}
 |0+1\rangle_{B}|(t_{2}-1)-t_{2}\rangle_{C},
\end{aligned}
\end{equation}
where $d_{1}=2k$, $s_{1}=k+2r_{1}-1$, $r_{1}=1,2,\cdots,\lfloor\frac{d_{2}-k}{2}\rfloor$, $s_{2}=k+2r_{2}-1$, $r_{2}=1,2,\cdots,\lfloor\frac{d_{3}-k}{2}\rfloor$. When $d_{2}-k$ is odd, $t_{1}=k+2r_{1}$; when $d_{2}-k$ is even, $t_{1}=k+2u_{1}$, $u_{1}=1,2,\cdots,\frac{d_{2}-k}{2}-1$. When $d_{3}-k$ is odd, $t_{2}=k+2r_{2}$; when $d_{3}-k$ is even, $t_{2}=k+2u_{2}$, $u_{2}=1,2,\cdots,\frac{d_{3}-k}{2}-1$.
\end{theorem}

\begin{proof}
We need to prove three features: the local distinguishability, local irredundancy and local activatability.

$(i)$ {\it Local distinguishability}. Define the following measurements: $M_{1}^{A}=\{P_{i}^{A}=P[|i-(2k-1-i)\rangle]_{A}, i=1,2,\cdots,k-1, P_{k}^{A}=I-\sum_{i=1}^{k-1}P_{i}^{A}\}$ for Alice, $M_{1}^{B}=\{P_{i}^{B}=P[|i\rangle]_{B}, i=1,\cdots,d_{2}-1, P_{d_{2}}^{B}=I-\sum_{i=1}^{d_{2}-1}P_{i}^{B}\}$ for Bob,  $M_{1}^{C}=\{P_{i}^{C}=P[|i\rangle]_{C}, i\in \mathbb{Z}_{d_{3}}\}$ for Charlie. The details of the discrimination protocol are given in Fig \ref{fig:2:even:d1d2d3} .
\begin{figure*}[h]
	\centering	\includegraphics[scale=0.5]{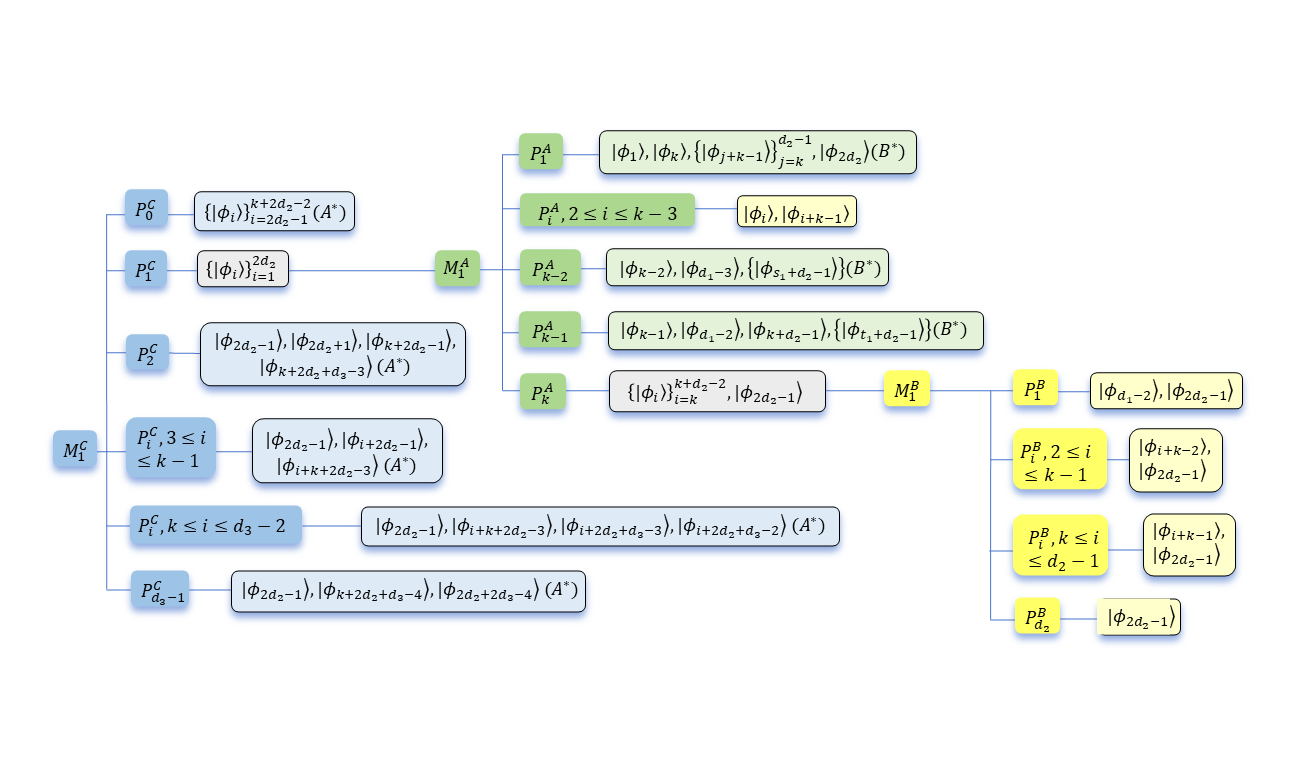}
	\caption{Distinction details for $\mathbb{C}^{d_{1}}\otimes\mathbb{C}^{d_{2}}\otimes\mathbb{C}^{d_{3}}$ ($d_{1}$ is even)}
         \label{fig:2:even:d1d2d3}
\end{figure*}

$(ii)$ {\it Local irredundancy}. Let $d_{1}=p_{1}^{(1)}p_{2}^{(1)}\cdots p_{n_{1}}^{(1)}$, where $p_{l}^{(1)}$ is a prime factor of $d_{1}$ (here we assume that $2= p_{1}^{(1)}\leq p_{2}^{(1)}\leq \cdots \leq p_{n_{1}}^{(1)}, n_{1}\geq2$). Suppose that  $\mathcal{H}_{A}=\mathbb{C}^{d_{1}}$ can be factored into an $n_{1}$-partite system $\otimes_{l=1}^{n_{1}}\mathcal{H}_{a_{l}}$ with $\mathcal{H}_{a_{l}}=\mathbb{C}^{p_{l}^{(1)}}$. The corresponding basis representation is $|[\sum_{l=1}^{n_{1}-1}(\lambda_{l} \prod_{j=l+1}^{n_{1}} p_{j}^{(1)})+\lambda_{n_{1}}]\rangle_{A}=|\lambda_{1}\lambda_{2}\cdots\lambda_{n_{1}}\rangle_{a_{1}a_{2}\cdots a_{n_{1}}}$, $\lambda_{l}\in \mathbb{Z}_{p_{l}^{(1)}}$. From this, it can be inferred that if $|\lambda\rangle_{A}=|\lambda_{1}\lambda_{2}\cdots\lambda_{n_{1}}\rangle_{a_{1}a_{2}\cdots a_{n_{1}}}$, then $|(2k-1-\lambda)\rangle_{A}=|(p_{1}^{(1)}-1-\lambda_{1})(p_{2}^{(1)}-1-\lambda_{2})\cdots(p_{n_{1}}^{(1)}-1-\lambda_{n_{1}})\rangle_{a_{1}a_{2}\cdots a_{n_{1}}}$. When Alice acts as the discarded party, this representation of $p$-ary numeral systems can quickly lock onto the target states, which will become nonorthogonal. Let $d_{i}=p_{1}^{(i)}p_{2}^{(i)}\cdots p_{n_{i}}^{(i)}$, where $p_{l}^{(i)}$ is a prime factor of $d_{i}$ (here we assume that $2\leq p_{1}^{(i)}\leq p_{2}^{(i)}\leq \cdots \leq p_{n_{i}}^{(i)}, n_{i}\geq2, i=2,3$).

Suppose $\mathcal{H}_{B}=\mathbb{C}^{d_{2}}$ can be factored into an $n_{2}$-partite system $\otimes_{l=1}^{n_{2}}\mathcal{H}_{b_{l}}$ with $\mathcal{H}_{b_{l}}=\mathbb{C}^{p_{l}^{(2)}}$, and
$\mathcal{H}_{C}=\mathbb{C}^{d_{3}}$ can be factored into an $n_{3}$-partite system $\otimes_{l=1}^{n_{3}}\mathcal{H}_{c_{l}}$ with $\mathcal{H}_{c_{l}}=\mathbb{C}^{p_{l}^{(3)}}$. Denote $|\phi_{s} \rangle=|\phi_{s} \rangle_{A}|\phi_{s} \rangle_{B}|\phi_{s} \rangle_{C}$ and the discarded subsystems by $\mathcal{T}_{A} \subseteq \{a_{1},a_{2}\cdots,a_{n_{1}}\}$, $\mathcal{T}_{B} \subseteq \{b_{1},b_{2},\cdots,b_{n_{2}}\}$ and $\mathcal{T}_{C} \subseteq \{c_{1},c_{2},\cdots,c_{n_{3}}\}$.
Consider the case $\mathcal{T}_{A}\neq\emptyset$. If Alice discards subsystem $a_{1}$, we have
$|\phi_{1}\rangle \to \frac{1}{4}[(|0\cdots 01\rangle \langle 0\cdots 01|+|(p_{2}^{(1)}-1)\cdots (p_{n_{1}-1}^{(1)}-1)(p_{n_{1}}^{(1)}-2)\rangle \langle (p_{2}^{(1)}-1)\cdots (p_{n_{1}-1}^{(1)}-1)(p_{n_{1}}^{(1)}-2)|)_{a_{2}\cdots a_{n_{1}-1}a_{n_{1}}} \otimes |0-1\rangle_{B} \langle 0-1|\otimes |1\rangle_{C} \langle 1|]$,
$|\phi_{k-2}\rangle \to \frac{1}{4}[(|(p_{2}^{(1)}-1)\cdots (p_{n_{1}-1}^{(1)}-1)(p_{n_{1}}^{(1)}-2)\rangle \langle (p_{2}^{(1)}-1)\cdots (p_{n_{1}-1}^{(1)}-1)(p_{n_{1}}^{(1)}-2)|+|0\cdots 01\rangle \langle 0\cdots 01|)_{a_{2}\cdots a_{n_{1}-1}a_{n_{1}}} \otimes |0-(k-2)\rangle_{B} \langle 0-(k-2)|\otimes |1\rangle_{C} \langle 1|]$, which will become nonorthogonal. Since partial trace operations preserve nonorthogonality, the two states remain nonorthogonal after discarding the subsystems $\mathcal{T}_{A}\cup \mathcal{T}_{B}\cup \mathcal{T}_{C}$ (where $a_{1}\in \mathcal{T}_{A}$).

Assume that the subsystem $a_{l}$ $(l=2,3,\cdots,n_{1}-1)$ is discarded. Let $|\upsilon^{(l)}\rangle=|00\cdots 0\cdots 01\rangle_{a_{1}a_{2}\cdots a_{l}\cdots a_{n_{1}-1}a_{n_{1}}}$, $|\omega^{(l)}\rangle=|00\cdots 1\cdots 01\rangle_{a_{1}a_{2}\cdots a_{l}\cdots a_{n_{1}-1}a_{n_{1}}}$. Then $|\phi_{\upsilon^{(l)}}\rangle \to \frac{1}{4}[(|00\cdots00\cdots 01\rangle \langle 00\cdots00\cdots 01|+|1(p_{2}^{(1)}-1)\cdots (p_{l-1}^{(1)}-1)(p_{l+1}^{(1)}-1)\cdots (p_{n_{1}-1}^{(1)}-1)(p_{n_{1}}^{(1)}-2)\rangle \langle 1(p_{2}^{(1)}-1)\cdots (p_{l-1}^{(1)}-1)(p_{l+1}^{(1)}-1)\cdots (p_{n_{1}-1}^{(1)}-1)(p_{n_{1}}^{(1)}-2)|)_{a_{1}a_{2}\cdots a_{l-1}a_{l+1}\cdots a_{n_{1}-1}a_{n_{1}}} \otimes |0-\upsilon^{(l)}\rangle_{B} \langle 0-\upsilon^{(l)}|\otimes |1\rangle_{C} \langle 1|]$ and
$|\phi_{\omega^{(l)}}\rangle \to \frac{1}{4}[(|00\cdots00\cdots 01\rangle \langle 00\cdots00\cdots 01|+|1(p_{2}^{(1)}-1)\cdots (p_{l-1}^{(1)}-1)(p_{l+1}^{(1)}-1)\cdots (p_{n_{1}-1}^{(1)}-1)(p_{n_{1}}^{(1)}-2)\rangle \langle 1(p_{2}^{(1)}-1)\cdots (p_{l-1}^{(1)}-1)(p_{l+1}^{(1)}-1)\cdots (p_{n_{1}-1}^{(1)}-1)(p_{n_{1}}^{(1)}-2)|)_{a_{1}a_{2}\cdots a_{l-1}a_{l+1}\cdots a_{n_{1}-1}a_{n_{1}}}\otimes |0-\omega^{(l)}\rangle_{B} \langle 0-\omega^{(l)}|\otimes |1\rangle_{C} \langle 1|]$ $(p_{l}^{(1)}\neq 3)$ or $\frac{1}{4}[|00\cdots00\cdots$ $01-1(p_{2}^{(1)}-1)\cdots (p_{l-1}^{(1)}-1)(p_{l+1}^{(1)}-1)\cdots (p_{n_{1}-1}^{(1)}-1)(p_{n_{1}}^{(1)}-2)\rangle_{a_{1}a_{2}\cdots a_{l-1}a_{l+1}\cdots a_{n_{1}-1}a_{n_{1}}} \langle 00\cdots00\cdots$ $ 01-1(p_{2}^{(1)}-1)\cdots (p_{l-1}^{(1)}-1)(p_{l+1}^{(1)}-1)\cdots (p_{n_{1}-1}^{(1)}-1)(p_{n_{1}}^{(1)}-2)| \otimes |0-\omega^{(l)}\rangle_{B} \langle 0-\omega^{(l)}|\otimes |1\rangle_{C} \langle 1|]$ $(p_{l}^{(1)}= 3)$. The two states remain nonorthogonal after discarding the subsystems $\mathcal{T}_{A}\cup \mathcal{T}_{B}\cup \mathcal{T}_{C}$ with $a_{l}\in \mathcal{T}_{A}$.

If Alice discards subsystem $a_{n_{1}}$, we have $|\phi_{k-2}\rangle \to \frac{1}{4}[(|0(p_{2}^{(1)}-1)\cdots (p_{n_{1}-1}^{(1)}-1)\rangle \langle 0(p_{2}^{(1)}-1)\cdots (p_{n_{1}-1}^{(1)}-1)|+|10\cdots 0\rangle \langle 10\cdots 0|)_{a_{1}a_{2}\cdots a_{n_{1}-1}} \otimes |0-(k-2)\rangle_{B} \langle 0-(k-2)|\otimes |1\rangle_{C} \langle 1|]$ $(p_{n_{1}}^{(1)}\neq 3)$ or $\frac{1}{4}[|0(p_{2}^{(1)}-1)\cdots (p_{n_{1}-1}^{(1)}-1)-10\cdots 0\rangle_{a_{1}a_{2}\cdots a_{n_{1}-1}} \langle 0(p_{2}^{(1)}-1)\cdots (p_{n_{1}-1}^{(1)}-1)-10\cdots 0| \otimes |0-(k-2)\rangle_{B} \langle 0-(k-2)|\otimes |1\rangle_{C} \langle 1|]$ $(p_{n_{1}}^{(1)}=3)$,
$|\phi_{k-1}\rangle \to \frac{1}{4}[(|0(p_{2}^{(1)}-1)\cdots (p_{n_{1}-1}^{(1)}-1)\rangle \langle 0(p_{2}^{(1)}-1)\cdots (p_{n_{1}-1}^{(1)}-1)|+|10\cdots 0\rangle \langle 10\cdots 0|)_{a_{1}a_{2}\cdots a_{n_{1}-1}} \otimes |0-(k-1)\rangle_{B} \langle 0-(k-1)|\otimes |1\rangle_{C} \langle 1|]$, which are nonorthogonal after discarding the subsystems $\mathcal{T}_{A}\cup \mathcal{T}_{B}\cup \mathcal{T}_{C}$ with $a_{n_{1}}\in \mathcal{T}_{A}$.

Now consider the case $\mathcal{T}_{A}=\emptyset$. We obtain  $\langle\phi_{s}|\phi_{t}\rangle_{A}\neq 0$ and $\langle\phi_{s}|\phi_{t}\rangle_{C}\neq 0$ for $s,t \in\{k,\cdots,k+d_{2}-2\}$. For $s,t\in\{d_{1}-2,k+2d_{2}-1,\cdots,k+2d_{2}+d_{3}-4\}$, we have $\langle\phi_{s}|\phi_{t}\rangle_{A}\neq 0$ and $\langle\phi_{s}|\phi_{t}\rangle_{B}\neq 0$. If the set is local redundancy, there exist nonempty sets among $\mathcal{T}_{B}$, $\mathcal{T}_{C}$ and some unitary matrices $U_{A}\in U(d_{1}), U_{B}\in U(d_{2})$, $U_{C}\in U(d_{3})$ such that $\{Tr_{\mathcal{T}_{B}\cup\mathcal{T}_{C}}[(U_{A}\otimes U_{B}\otimes U_{C})|\phi_{s}\rangle \langle\phi_{s}|(U_{A}^{\dagger}\otimes U_{B}^{\dagger}\otimes U_{C}^{\dagger})]\}_{s=1}^{2d_{2}+2d_{3}-4}$ are pairwise orthogonal, which is equivalent to the orthogonality of the set $\{(U_{A}|\phi_{s}\rangle_{A}\langle\phi_{s}|U_{A}^{\dagger})
\otimes Tr_{\mathcal{T}_{B}}(U_{B}|\phi_{s}\rangle_{B}\langle\phi_{s}|U_{B}^{\dagger})
\otimes Tr_{\mathcal{T}_{C}}(U_{C}|\phi_{s}\rangle_{C}\langle\phi_{s}|U_{C}^{\dagger})
\}_{s=1}^{2d_{2}+2d_{3}-4}$. Since the partial trace operation preserves nonorthogonality, the following sets fulfill the conditions: $\{Tr_{\mathcal{T}_{B}}(U_{B}|\phi_{s}\rangle_{B} \langle\phi_{s}|U_{B}^{\dagger})\}_{s=k}^{k+d_{2}-2}$, $\{Tr_{\mathcal{T}_{C}}(U_{C}|\phi_{s}\rangle_{C} \langle\phi_{s}|U_{C}^{\dagger}), s=d_{1}-2,k+2d_{2}-1,\cdots,k+2d_{2}+d_{3}-4\}$. Without loss of generality, we assume that $\mathcal{T}_{B}\neq\emptyset$. Therefore, the resulting states $\{Tr_{\mathcal{T}_{B}}(U_{B}|\phi_{s}\rangle_{B} \langle\phi_{s}|U_{B}^{\dagger})\}_{s=k}^{k+d_{2}-2}$ are pairwise orthogonal in the systems corresponding to $\{b_{1},b_{2},\cdots,b_{n_{2}}\}\setminus\mathcal{T}_{B}$ whose dimension is at most $\frac{d_{2}}{2}<d_{2}-1$. Hence we arrive at a contradiction.

$(iii)$ {\it Genuine nonlocality can be activated}. Suppose that Alice performs the measurement $K^{A}=\{ K_{1}^{A}=\sum_{j=0}^{k-1}|j\rangle_{A} \langle j|, K_{2}^{A}=\sum_{j=k}^{2k-1}|j\rangle_{A} \langle j|\}$. Either the $K_{1}^{A}$
or $K_{2}^{A}$ clicks, by Proposition \ref{prop:gn:d1d2d3}, the set will be converted into a type-\uppercase\expandafter{\romannumeral 2} genuinely nonlocal set. This completes the proof.
\end{proof}

\subsection{$\mathbb{C}^{d_{1}}\otimes\mathbb{C}^{d_{2}}\otimes\mathbb{C}^{d_{3}}$ ($d_{1}$ is odd)}
\begin{theorem}\label{th:2:d1d2d3:odd}
The type-\uppercase\expandafter{\romannumeral 2} genuine nonlocality of the following $2d_{2}+2d_{3}-4$ orthogonal product states in $\mathbb{C}^{d_{1}}\otimes\mathbb{C}^{d_{2}}\otimes\mathbb{C}^{d_{3}}$ ($5\leq \frac{d_{1}+1}{2}\leq d_{2}\leq d_{3}$, $d_{1}$ is odd) can be activated via~{\textcolor{blue}{OPLMs by Alice}}:
\begin{equation*}
\begin{aligned}
|\phi_{1}\rangle=&|(1+k)-(k-2)\rangle_{A}|0-1\rangle_{B}|1\rangle_{C},\\
|\phi_{2}\rangle=&|(2+k)-(k-1)\rangle_{A}|0-2\rangle_{B}|1\rangle_{C},\\
|\phi_{i}\rangle=&|(i+k)-(k-i)\rangle_{A}|0-i\rangle_{B}|1\rangle_{C},3\leq i\leq k,\\
|\phi_{1+k}\rangle=&|k-(1+k)+(k-2)-1\rangle_{A}|2\rangle_{B}|1\rangle_{C},\\
|\phi_{i+k}\rangle=&|k-(i+k)+(k-i)-(k-1)\rangle_{A}|j\rangle_{B}\\&
|1\rangle_{C},~2\leq i\leq k-1,j=i+1;i=k,j=1,\\
|\phi_{j+k}\rangle=&|k-(1+k)+(k-2)-(k-1)\rangle_{A}|j\rangle_{B}\\&
|1\rangle_{C},~k+1\leq j\leq d_{2}-1,\\
|\phi_{k+d_{2}}\rangle=&|2k-0\rangle_{A}|2-(k+1)\rangle_{B}|1\rangle_{C},~k+1<d_{2},\\
|\phi_{s_{1}+d_{2}-1}\rangle=&|(2k-1)-1\rangle_{A}|(s_{1}-1)-s_{1}\rangle_{B}|1\rangle_{C},\\
|\phi_{t_{1}+d_{2}-1}\rangle=&|2k-0\rangle_{A}|(t_{1}-1)-t_{1}\rangle_{B}|1\rangle_{C},\\
|\phi_{2d_{2}-1}\rangle=&|_{0}+_{(d_{1}-1)}\rangle_{A}|_{0}+_{(d_{2}-1)}\rangle_{B}|_{0}+_{(d_{3}-1)}\rangle_{C},\\
|\phi_{2d_{2}}\rangle=&|(1+k)-(k-2)\rangle_{A}|0+1\rangle_{B}|0-1\rangle_{C},\\
|\phi_{2d_{2}+1}\rangle=&|(2+k)-(k-1)\rangle_{A}|0+1\rangle_{B}|0-2\rangle_{C},\\
|\phi_{i+2d_{2}-1}\rangle=&|(i+k)-(k-i)\rangle_{A}|0+1\rangle_{B}|0-i\rangle_{C},\\&3\leq i\leq k,\\
|\phi_{k+2d_{2}}\rangle=&|k-(k+1)+(k-2)-1\rangle_{A}|0+1\rangle_{B}|2\rangle_{C},\\
|\phi_{i+k+2d_{2}-1}\rangle=&|k-(k+i)+(k-i)-(k-1)\rangle_{A}\\
\end{aligned}
\end{equation*}
\begin{equation}
\begin{aligned}
&|0+1\rangle_{B}|j\rangle_{C},~2\leq i\leq k-1,j=i+1,\\
|\phi_{j+k+2d_{2}-2}\rangle=&|k-(k+1)+(k-2)-(k-1)\rangle_{A}\\&|0+1\rangle_{B}|j\rangle_{C},~k+1\leq j\leq d_{3}-1,\\
|\phi_{k+2d_{2}+d_{3}-2}\rangle=&|2k-0\rangle_{A}|0+1\rangle_{B}|2-(k+1)\rangle_{C},\\&k+1<d_{3},\\
|\phi_{s_{2}+2d_{2}+d_{3}-3}\rangle=&|(2k-1)-1\rangle_{A}|0+1\rangle_{B}|(s_{2}-1)-s_{2}\rangle_{C},\\
|\phi_{t_{2}+2d_{2}+d_{3}-3}\rangle=&|2k-0\rangle_{A}|0+1\rangle_{B}|(t_{2}-1)-t_{2}\rangle_{C},
\end{aligned}
\end{equation}
where $d_{1}=2k+1$, $s_{1}=k+2r_{1}$, $r_{1}=1,2,\cdots,\lfloor\frac{d_{2}-k-1}{2}\rfloor$, $s_{2}=k+2r_{2}$, $r_{2}=1,2,\cdots,\lfloor\frac{d_{3}-k-1}{2}\rfloor$. When $d_{2}-k$ is even, $t_{1}=k+2r_{1}+1$; when $d_{2}-k$ is odd, $t_{1}=k+2u_{1}+1$, $u_{1}=1,2,\cdots,\frac{d_{2}-k-1}{2}-1$. When $d_{3}-k$ is even, $t_{2}=k+2r_{2}+1$; when $d_{3}-k$ is odd, $t_{2}=k+2u_{2}+1$, $u_{2}=1,2,\cdots,\frac{d_{3}-k-1}{2}-1$.
\end{theorem}

The proof of Theorem 4 is similar to Theorem 3.

\section{Conclusion}\label{sec:6}
In this manuscript, several structures of locally distinguishable orthogonal product states are presented, whose genuine nonlocality can be activated by orthogonality-preserving local
measurements. In Ref.\cite{Bandyopadhyay2021} an open problem was raised: So whether the activation phenomenon can also be observed in other manifestations of nonlocality is an intriguing question. Here we have found the activation phenomenon in genuine nonlocality that is stronger than nonlocality. We have first constructed $2d_{2}+2d_{3}-4$ orthogonal product states with type-\uppercase\expandafter{\romannumeral 1} genuine nonlocality in $\mathbb{C}^{d_{1}}\otimes\mathbb{C}^{d_{2}}\otimes\mathbb{C}^{d_{3}}$~($4\leq d_{1}\leq d_{2}\leq d_{3}$). Compared with Ref.\cite{Rout2021}, the cardinality of this construction is significantly reduced. Then we activated the type-I genuine nonlocality of distinguishable orthogonal product state sets in $\mathbb{C}^{d_{1}}\otimes\mathbb{C}^{d_{2}}\otimes\mathbb{C}^{d_{3}}$~($11\leq d_{1}\leq d_{2}\leq d_{3}$, $d_{i}$ is odd) and $\mathbb{C}^{d_{1}}\otimes\mathbb{C}^{d_{2}}\otimes\mathbb{C}^{d_{3}}$~($8\leq d_{1}\leq d_{2}\leq d_{3}$, $d_{i}$ is even). For the type-\uppercase\expandafter{\romannumeral 2} genuine nonlocality, we have put forward activatable sets without entanglement in $\mathbb{C}^{d_{1}}\otimes\mathbb{C}^{d_{2}}\otimes\mathbb{C}^{d_{3}}$~(when $d_{1}$ is odd, $5\leq \frac{d_{1}+1}{2}\leq d_{2}\leq d_{3}$; when $d_{1}$ is even, $4\leq \frac{d_{1}}{2}\leq d_{2}\leq d_{3}$). In particular, we have applied $p$-ary numeral systems to simplify the proof of local irredundancy for even-dimensional subsystems.

This work focuses on the activation of distinguishable orthogonal quantum state sets under \textit{single-party local operations}. It would be also interesting to explore the activation of nonlocality under \textit{multiparty joint operations}. Moreover, we have investigated primarily the activation of \textit{genuine nonlocality}. The activation of other variants such as \textit{strong nonlocality} and \textit{strongest nonlocality} for local indistinguishability warrants further investigation.

\section*{Acknowledgments}

This work is supported by the NSFC (Grants Nos. 12571493 and 62272208), the Interdisciplinary Research Fundation of Hebei Normal University (Grant No. L2025J01), the Basic Research Project of Shijiazhuang Municipal Universities in Hebei Province (Grant No. 241790697A), the Open Foundation of State key Laboratory of Networking and Switching Technology (Beijing University of Posts and Telecommunications) (SKLNST-2025-1-19).

\appendix

\section{Proof of Proposition \ref{prop:gn:d1d2d3}} \label{app:prop:gn:d1d2d3}

\begin{proof}
Consider the subset $\{|\phi_{s}\rangle\} _{s=1}^{2d_{2}-1}$: Alice and Bob share a locally indistinguishable set in $\mathbb{C}^{d_{1}}\otimes\mathbb{C}^{d_{2}}$, while in Charlie's side are nonorthogonal. As a result, even if Charlie comes together with either Alice or Bob, it is not possible to perfectly distinguish these states, leading to the local indistinguishability in partitions $B|CA$ and $A|BC$.

Similarly, the set $\{|\phi_{2d_{1}-2}\rangle \bigcup \{ |\phi_{s}\rangle\} _{s=2d_{2}-1}^{2d_{2}+2d_{3}-4}\}$ is a locally indistinguishable set in $\mathbb{C}^{d_{1}}\otimes\mathbb{C}^{d_{3}}$ shared between Alice and Charlie, the members are nonorthogonal in Bob's side. Then, even if Bob comes together with either Alice or Charlie, it is not possible to perfectly distinguish these states, leading to the local indistinguishability in partitions $C|AB$ and $A|BC$.

Moreover, the above set is locally irreducible when all the parties are separated, since they can only perform trivial OPLMs. Hence, it is genuinely nonlocal of type \uppercase\expandafter{\romannumeral 2}. This completes the proof.
\end{proof}

\section{Proof of Lemma \ref{lemma}} \label{app:lemma}

\begin{proof}
We need to prove that all orthogonality-preserving measurements performed by Alice and Bob must be trivial.

If Alice makes the first measurements with the POVM elements given by $E_{1}=(a_{i,j})_{i,j\in\mathbb{Z}_{d_{1}}}$, we have $\langle\phi_{i}|E_{1}\otimes I_{2}|\phi_{j}\rangle=0$  for $i\neq j$. From Table \ref{tab:Alice}, we observe that $E_{1}\propto I$, indicating that Alice can only perform trivial measurements.
\begin{table}[htbp]\small
  \newcommand{\tabincell}[2]{\begin{tabular}{@{}#1@{}}#2\end{tabular}}
  \centering
  \begin{tabular}{lll}
  \toprule
  \hline
  \specialrule{0em}{1.5pt}{1.5pt}
    Pair of states~~&~~Key entries~~~&~~Value range~~~\\
  \specialrule{0em}{1.5pt}{1.5pt}
  \midrule
  \specialrule{0em}{1.5pt}{1.5pt}
    \tabincell{c}{$|\phi_{1}\rangle,|\phi_{2}\rangle$}&~~\tabincell{l}{$a_{(d_{1}-2),(d_{1}-1)}=0$\\$a_{(d_{1}-1),(d_{1}-2)}=0$}&~~\tabincell{c}{null}\\
    \specialrule{0em}{3.5pt}{3.5pt}
    \tabincell{c}{$|\phi_{1}\rangle,|\phi_{i}\rangle$}&~~\tabincell{l}{$a_{(d_{1}-2),(d_{1}-i)}=0$\\$a_{(d_{1}-i),(d_{1}-2)}=0$}&~~\tabincell{c}{$3\leq i\leq d_{1}$}\\
    \specialrule{0em}{3.5pt}{3.5pt}
    \tabincell{c}{$|\phi_{2}\rangle,|\phi_{i}\rangle$}&~~\tabincell{l}{$a_{(d_{1}-1),(d_{1}-i)}=0$\\$a_{(d_{1}-i),(d_{1}-1)}=0$}&~~\tabincell{c}{$3\leq i\leq d_{1}$}\\
    \specialrule{0em}{3.5pt}{3.5pt}
    \tabincell{c}{$|\phi_{i}\rangle,|\phi_{i^{'}}\rangle$}&~~\tabincell{l}{$a_{(d_{1}-i),(d_{1}-i^{'})}=0$\\$a_{(d_{1}-i^{'}),(d_{1}-i)}=0$}&~~\tabincell{c}{$3\leq i\neq i^{'}\leq d_{1}$}\\
    \specialrule{0em}{3.5pt}{3.5pt}
    \tabincell{c}{$|\phi_{i+d_{1}}\rangle,|\phi_{2d_{2}-1}\rangle$}&~~\tabincell{l}{~~~$a_{(d_{1}-i),(d_{1}-i)}$\\$=a_{(d_{1}-1),(d_{1}-1)}$}&~~\tabincell{c}{$2\leq i\leq d_{1}$}\\
    \specialrule{0em}{1.5pt}{1.5pt}
    \hline
    \bottomrule
  \end{tabular}
  \caption{The entries in operator $E_{1}$}
  \label{tab:Alice}
\end{table}

Similarly, if Bob performs local measurements with the POVM elements given by $E_{2}=(b_{i,j})_{i,j\in\mathbb{Z}_{d_{2}}}$, we have $\langle\phi_{i}|I_{1}\otimes E_{2}|\phi_{j}\rangle=0$ for $i\neq j$. From the Table \ref{tab:Bob}, we conclude that Bob's operators must be trivial too.
\begin{table}[htbp]\small
  \newcommand{\tabincell}[2]{\begin{tabular}{@{}#1@{}}#2\end{tabular}}
  \centering
  \begin{tabular}{lll}
  \toprule
  \hline
  \specialrule{0em}{1.5pt}{1.5pt}
    Pair of states~~&Key entries~~~&Value range~~~\\
  \specialrule{0em}{1.5pt}{1.5pt}
  \midrule
  \specialrule{0em}{1.5pt}{1.5pt}
    \tabincell{c}{$|\phi_{1+d_{1}}\rangle,|\phi_{j+d_{1}}\rangle$}&\tabincell{c}{$b_{2,j}=b_{j,2}=0$}&\tabincell{c}{$d_{1}+1\leq j\leq d_{2}-1$}\\
    \specialrule{0em}{3.5pt}{3.5pt}
    \tabincell{c}{$|\phi_{1+d_{1}}\rangle,|\phi_{2+d_{1}}\rangle$}&\tabincell{c}{$b_{2,3}=b_{3,2}=0$}&\tabincell{c}{null}\\
    \specialrule{0em}{3.5pt}{3.5pt}
    \tabincell{c}{$|\phi_{1+d_{1}}\rangle,|\phi_{2d_{1}-1}\rangle$}&\tabincell{c}{$b_{2,d_{1}}=b_{d_{1},2}=0$}&\tabincell{c}{null}\\
    \specialrule{0em}{3.5pt}{3.5pt}
    \tabincell{c}{$|\phi_{i+d_{1}}\rangle,|\phi_{i^{'}+d_{1}}\rangle$}&\tabincell{l}{$b_{(i+1),(i^{'}+1)}=0$\\$b_{(i^{'}+1),(i+1)}=0$}&\tabincell{c}{$2\leq i\neq i^{'}\leq d_{1}-1$}\\
    \specialrule{0em}{3.5pt}{3.5pt}
    \tabincell{c}{$|\phi_{i+d_{1}}\rangle,|\phi_{2d_{1}}\rangle$}&\tabincell{l}{$b_{(i+1),1}=0$\\$b_{1,(i+1)}=0$}&\tabincell{c}{$2\leq i\leq d_{1}-1$}\\
     \specialrule{0em}{3.5pt}{3.5pt}
    \tabincell{c}{$|\phi_{i+d_{1}}\rangle,|\phi_{j+d_{1}}\rangle$}&\tabincell{l}{$b_{(i+1),j}=0$\\$b_{j,(i+1)}=0$}&\tabincell{l}{$2\leq i\leq d_{1}-1,$\\$d_{1}+1\leq j\leq d_{2}-1$}\\
     \specialrule{0em}{3.5pt}{3.5pt}
    \tabincell{c}{$|\phi_{2d_{1}}\rangle,|\phi_{j+d_{1}}\rangle$}&\tabincell{c}{$b_{1,j}=b_{j,1}=0$}&\tabincell{c}{$d_{1}+1\leq j\leq d_{2}-1$}\\
     \specialrule{0em}{3.5pt}{3.5pt}
    \tabincell{c}{$|\phi_{j+d_{1}}\rangle,|\phi_{j^{'}+d_{1}}\rangle$}&\tabincell{c}{$b_{j,j^{'}}=b_{j^{'},j}=0$}&\tabincell{l}{$d_{1}+1\leq j\neq j^{'}$\\$\leq d_{2}-1$}\\
     \specialrule{0em}{3.5pt}{3.5pt}
    \tabincell{c}{$|\phi_{i}\rangle,|\phi_{i+d_{1}}\rangle$}&\tabincell{l}{$b_{0,(i+1)}=0$\\$b_{(i+1),0}=0$}&\tabincell{c}{$2\leq i\leq d_{1}-1$}\\
     \specialrule{0em}{3.5pt}{3.5pt}
    \tabincell{c}{$|\phi_{d_{1}}\rangle,|\phi_{2d_{1}}\rangle$}&\tabincell{c}{$b_{0,1}=b_{1,0}=0$}&\tabincell{c}{null}\\
     \specialrule{0em}{3.5pt}{3.5pt}
    \tabincell{c}{$|\phi_{2}\rangle,|\phi_{i+d_{1}}\rangle$}&\tabincell{l}{$b_{2,(i+1)}=0$\\$b_{(i+1),2}=0$}&\tabincell{c}{$3\leq i\leq d_{1}-1$}\\
     \specialrule{0em}{3.5pt}{3.5pt}
    \tabincell{c}{$|\phi_{2}\rangle,|\phi_{2d_{1}}\rangle$}&\tabincell{c}{$b_{1,2}=b_{2,1}=0$}&\tabincell{c}{null}\\
     \specialrule{0em}{3.5pt}{3.5pt}
    \tabincell{c}{$|\phi_{2}\rangle,|\phi_{j+d_{1}}\rangle$}&\tabincell{c}{$b_{0,j}=b_{j,0}=0$}&\tabincell{c}{$d_{1}+1\leq j\leq d_{2}-1$}\\
     \specialrule{0em}{3.5pt}{3.5pt}
    \tabincell{c}{$|\phi_{1}\rangle,|\phi_{1+d_{1}}\rangle$}&\tabincell{c}{$b_{0,2}=b_{2,0}=0$}&\tabincell{c}{null}\\
     \specialrule{0em}{3.5pt}{3.5pt}
    \tabincell{c}{$|\phi_{i}\rangle,|\phi_{2d_{2}-1}\rangle$}&\tabincell{c}{$b_{0,0}=b_{i,i}$}&\tabincell{c}{$1\leq i\leq d_{1}$}\\
     \specialrule{0em}{3.5pt}{3.5pt}
    \tabincell{c}{$|\phi_{d_{1}+d_{2}}\rangle,|\phi_{2d_{2}-1}\rangle$}&\tabincell{c}{$b_{2,2}=b_{(d_{1}+1),(d_{1}+1)}$}&\tabincell{c}{null}\\
     \specialrule{0em}{3.5pt}{3.5pt}
    \tabincell{c}{$|\phi_{s_{1}+d_{2}-1}\rangle,|\phi_{2d_{2}-1}\rangle$}&\tabincell{c}{$b_{(s_{1}-1),(s_{1}-1)}=b_{s_{1},s_{1}}$}&\tabincell{c}{proposition 1}\\
     \specialrule{0em}{3.5pt}{3.5pt}
    \tabincell{c}{$|\phi_{t_{1}+d_{2}-1}\rangle,|\phi_{2d_{2}-1}\rangle$}&\tabincell{c}{$b_{(t_{1}-1),(t_{1}-1)}=b_{t_{1},t_{1}}$}&\tabincell{c}{proposition 1}\\
    \specialrule{0em}{1.5pt}{1.5pt}
    \hline
    \bottomrule
  \end{tabular}
  \caption{The entries in operator $E_{2}$}
  \label{tab:Bob}
\end{table}
\end{proof}


\end{document}